\documentclass[11pt, oneside]{article}   	% use "amsart" instead of "article" for AMSLaTeX format
\usepackage{geometry}                		% See geometry.pdf to learn the layout options. There are lots.
\usepackage{graphicx}				% Use pdf, png, jpg, or eps§ with pdflatex; use eps in DVI mode
\usepackage{subfig}
\usepackage{amssymb}
\usepackage{float}
\usepackage{amssymb}
\usepackage{amsmath}
\usepackage{amsthm}
\usepackage{color}
\usepackage{graphicx}
\usepackage{hyperref}
\usepackage{fullpage}
\usepackage{graphicx}
\usepackage{mathtools}
\usepackage{microtype}
\usepackage{bbm}
\usepackage{booktabs}

\usepackage[linesnumbered,algoruled,titlenotnumbered,vlined]{algorithm2e}
\DontPrintSemicolon

\renewenvironment{proof}{\vspace{-0.05in}\noindent{\bf Proof:}}%
        {\hspace*{\fill}$\Box$\par}
\newenvironment{proofof}[1]{\smallskip\noindent{\bf Proof of #1:}}%
        {\hspace*{\fill}$\Box$\par}
        {\hspace*{\fill}$\Box$\par}

\usepackage{algorithmic}

\usepackage{lmodern}
\usepackage[T1]{fontenc}
\usepackage{multirow}

\newtheorem{theorem}{Theorem}
\newtheorem{lemma}[theorem]{Lemma}
\newtheorem{definition}[theorem]{Definition}
\newtheorem{fact}[theorem]{Fact}

\newtheorem{observation}[theorem]{Observation}
\newtheorem{proposition}[theorem]{Proposition}
\newtheorem{assumption}[theorem]{Assumption}

\newcommand{\ie}{{\em i.e.}}

\DeclareMathOperator*{\var}{var}
\DeclareMathOperator*{\cov}{cov}
\DeclareMathOperator*{\io}{i.o.}

\title{On the Asymptotic Optimality of Work-Conserving Disciplines in Completion Time Minimization}
\author{Wenxin Li \\
 The Ohio State University\\
{\tt wenxinliwx.1@gmail.com}\\
{\tt li.7328@osu.edu}\\
}

\begin{document}
\maketitle

\begin{abstract}
In this paper, we prove that under mild stochastic assumptions,  work-conserving disciplines are asymptotic optimal for minimizing total completion time. As a byproduct of our analysis, we obtain tight upper bound on the competitive ratios of work-conserving disciplines on minimizing the metric of flow time.
%In other words, the stochastic system driven by the arrival and service processes is stable.

\end{abstract}

%Under the $3$-field notation introduced by Graham~\cite, our result
%in deterministic models
\section{Introduction}
Minimizing the (weighted) total completion time, one of the most basic performance metric in scheduling theory, has been extensively studied since the 1990s~\cite{phillips1998minimizing}, and the earliest study can be traced back to 1950s~\cite{smith1956various}. Formally, we are given a set of $n$ jobs $N=\{1, 2,\ldots, n\}$, each job has a workload of $p_{j}$. Let $C_{j}$ denote the completion time of job $j$, the goal is to find a schedule that minimizes the total (weighted) completion time $\sum_{j\in [n]}{C_{j}}$.

The most basic problem in this context is the single machine model with batch arrivals,  (\ie, $1||\sum_{j}C_{j}$ in the standard $3$-field notation introduced by Graham et al.~\cite{graham1979optimization}), which can be exactly solved by the Shortest Processing Time (SPT). There are numerous generalizations of this classic formulation, including the setting with multiple machines, precedence constraints and release dates~\cite{chekuri2004approximation}. Almost all but a few relatively simple variants under consideration are NP-hard, for which various efficient offline approximation algorithms are available~\cite{chekuri2004approximation, chekuri2001approximation}. Recently there has been a line of work on improving the approximation guarantee for total weighted completion time objective~\cite{bansal2019lift,im2016better,skutella20162,li2017scheduling}. The corresponding online setting is also an active area of research, in which jobs arrive online and each job becomes known to the algorithm only after its arrival. For instance, Anderson and Potts~\cite{anderson2002line} considered the problem of minimizing the weighted completion time in the non-preemptive single machine model, and proved that a simple modification of the shortest weighted processing time rule achieves the optimal competitive ratio of two.  Shmoys et al.~\cite{shmoys1995scheduling} showed how to obtain a $2\rho$-competitive online non-clairvoyant algorithm from an offline $\rho$-approximation algorithm. In a similar flavor, Hall et al.~\cite{hall1997scheduling} presented a technique for converting a $\rho$-approximation algorithm of \emph{the maximum scheduled weighted problem} to a $4\rho$-competitive algorithm for completion minimization.

To compare the performance of different disciplines, deterministic models always consider the worst possible input, which does not correspond to any inherent properties of the input. In addition to the aforementioned results in the deterministic setting, there has also been a considerate amount of work on stochastic models, which helps to explain the empirical performance. Specifically, there is a line of work that utilizes \emph{asymptotic analysis} to evaluate system performance in a large scale, often with certain
%statistical
stochastic assumptions on the input data. Chou et al.~\cite{chou2006asymptotic} studied the weighted completion time minimization problem with release dates in single machine model, and proved that the expected weighted completion time under the non-preemptive weighted shortest expected processing time among available jobs (WSEPTA) algorithm is asymptotically optimal when the number of jobs increases to infinity, if job workload and weights are bounded and the job workload are mutually independent random variables. Kaminsky and Simchi-Levi~\cite{kaminsky2001asymptotic} proved the asymptotic optimality of the SPT for the total completion time objective in flow shop model, where each job must be sequentially processed on the machines and every job has the same routing.
%Besides
%related to optimizing completion time
%Many variants of these scheduling problems are NP-hard, various efficient offline approximation algorithms and online algorithms are proposed~\cite{}.

%with release dates and stochastic processing time and focus on the objective of minimizing the weighted completion times of all jobs.
It is observed that the metric of completion time is more robust with respect to various changes in the input instances~\cite{bansal}, compared the flow time objective. In addition, from all the mentioned results above, we can see that although the suggested approaches for completion time optimization are ad-hoc, various different scheduling disciplines all admit desirable performance guarantee. For example, the seminal list-scheduling algorithm~\cite{graham1969bounds}
% R. L. Graham. Bounds on multiprocessing timing anomalies. SIAM JOURNAL ON APPLIED MATHEMATICS, 17(2):416–429, 1969.
achieves a constant gap of two to the optimal, even in the worst case scenario, which is usually overly pessimistic, while WSEPTA and SPT are asymptotically optimal in the stochastic model. Collectively, these observations lead to the question of \emph{whether there is a unifying characterization or explanation on the excellent performance of a certain class of scheduling disciplines in different settings.
%of scheduling disciplines that can achieve completion time optimality.
}

%Previous suggested approaches for completion time optimization are ad-hoc

Our first main result answers this question in affirmative, which is formally stated in Theorem~\ref{completiontheo}, and can be summarized in words as,
\begin{quote}
\emph{As long as machines are kept busy whenever possible, the total job completion time are optimum when the number of jobs is sufficiently large.}
\end{quote}
To show this result, tight competitive ratio bounds for work-conserving disciplines on flow time are established, which is summarized in Table~\ref{wctable}.

\paragraph{Related Work}\label{relatedworksec}
For minimizing total flow time, Zheng et al.~\cite{zheng2013new,zheng2012performance} considered  Map Reduce model and proved that with whole probability, any work-conserving scheduling algorithms have bounded gap with respect to the optimal algorithm, which holds under bounded job size assumption and stochastic assumptions on the input data. 

The most relevant to our work is~\cite{chen2007probabilistic}, which provides the optimality condition when the job workloads are upper bounded by a constant, together with the additional assumption that job workload and interarrival time are i.i.d distributed. However, it is natural to expect the maximum job length to be unbounded when the number of jobs increases to infinity. Moreover, the input data cannot assumed to be identical distributed in every situation. Our objective is to provide a deeper understanding of the completion time metric, which is potentially useful to identify disciplines that are both effective (in minimizing completion time) and easy to implement. 

%related work in Section~\ref{relatedworksec}. 

\subsection{Main contributions}\label{secmaincon}
\paragraph{Asymptotic optimality in minimizing completion time.} The appeal of our main result is that the assumptions are fundamentally natural and general, we does not require identical distribution assumption or assume any specific distributions on the input. 

\begin{theorem}\label{completiontheo}	
Under assumption \ref{completiontimeassump} and ~\ref{stableassump}, any work conserving algorithm $\pi$ is  almost surely asymptotically optimal for online completion time minimization problems $Pm|r_{j}|\sum_{j}C_{j}$, $Pm|r_{j}, pmtn|\sum_{j}C_{j}$, $Qm|r_{j}|\sum_{j}C_{j}$, $Qm|r_{j}, pmtn|\sum_{j}C_{j}$, \ie,
\begin{align}\label{completiomcon}
\lim_{n\rightarrow \infty} \frac{\sum_{i\in [n]}{C^{\pi}_{i}}}{\sum_{i\in [n]}C^{\pi^{*}}_{i}}=1\;(\forall \pi\in \Pi_{\mathcal{W}})
\end{align}
holds almost surely for any input instance, where $\Pi_{\mathcal{W}}$ denotes the class of work-conserving disciplines.
\end{theorem}

\begin{assumption}\label{completiontimeassump}
Job workload $\{p_{i}\}_{i\in [n]}$ and interarrival time $\{\Delta r_{i}\}_{i\in [n]}$ are independently distributed, the $(2+\epsilon)$-th moment of job workload and the second moment of interarrival time are finite, and the mean values of interarrival time are identical.
\end{assumption}

\begin{assumption}\label{stableassump}
The stochastic system driven by the arrival and service processes is stable.
\end{assumption}

In addition, our result can be further applied in the following settings.
\begin{itemize}
%\item \emph{Preemptive and non-preemptive model.} In preemptive model, the job that is running can be interrupted and later continued on any machines. Otherwise the system must follow the ``run to completion'' rule in the non-preemptive setting.
\item \emph{Constant number of interjob precedence phase.} Interjob precedence constraint $i\rightarrow j$ implies that job $i$ must be finished before we start to process job $j$. We assume that there are constant number of such precedence constraints. 
\item \emph{Multitask job with arbitrary intertask precedence constraint.} Each job consists of multiple tasks, the job is considered to be completed until all its tasks are finished. The precedence constraints between tasks within the same job can be arbitrary.
\end{itemize}

\paragraph{Tight competitive ratio bound in flow time minimization. }
To the best of our knowledge, this is the first tight characterizations on the worst case performance of work-conserving scheduling algorithms. In Theorem~\ref{wclowbound} we show that the total flow time under any work-conserving algorithms is always no more than $2B$ times that under the optimal algorithm. Here parameter $B=p_{\max}/p_{\min}$ represents the ratio of the maximum to the minimum job workload. In addition, this competitive ratio upper bound is shown to be tight up to a constant. On the negative side, any non-preemptive scheduling algorithm cannot achieve a competitive ratio better than $B^{1-\varepsilon}$ for any given constant $\varepsilon>0$. Together with the well-known competitive ratio lower bound for preemptive scenario, our result is summarized in Table~\ref{wctable}.
%Our result suggests that non-preemptive work-scheduling algorithms have almost the same (up to lower order terms) worst case performance, which can be regarded as ???.

\begin{table}[H]
\centering
\begin{tabular}{|l|l|l|l|l}
\cline{1-4}
\multicolumn{2}{|l|}{Scenario}  & Competitive Ratio Supremum                  &  Infimum  &  \\ \cline{1-4}
\multirow{2}{*}{Preemptive}   & Single machine & \multirow{2}{*}{$\Theta(B)$ (\textbf{Theorem~\ref{tightuppinfworkcon}} )} & $1~\cite{leonardi1997approximating}$  &  \\ \cline{2-2} \cline{4-4}
                     & Multiple machine &                    & $\Theta(\log B)$~\cite{leonardi1997approximating} &  \\ \cline{1-4}
\multicolumn{2}{|l|}{Non-preemptive} & $\Theta(B)$ (\textbf{Theorem~\ref{tightuppinfworkcon}} )                 & $\Omega(B^{1-\varepsilon})$~\cite{leonardi1997approximating} &  \\ \cline{1-4}
\end{tabular}
\caption{Summarization of worst case performance of work-conserving algorithms}
\label{wctable}
\end{table}

%Our main contributions are presented in Section~\ref{secmaincon}. 
The rest of this paper is organized as follows. Section~\ref{modeldefsec} describes model and definitions. We prove the competitive ratio bound and optimality condition in Section~\ref{tightflowboundsec} and Section~\ref{completiontime} respectively. Possible generalizations and numerical results are given in Section~\ref{secgeneralization} and \ref{numericalsec}. The paper in concluded in Section~\ref{conclusionsec}.

%Moreover, an online scheduler called Available Shortest Remaining Processing Time (ASRPT) is presented and shown to achieve an efficiency ratio no more than two. The aforementioned results are derived in the MapReduce framework, and indeed also hold in the basic parallel machine setting. \cite{zheng2015exploiting} further proves that in the large system limit, work-conserving algorithms are optimal in preemptive and parallelizeable 
%Beyond sheding light on , a unify ?? will allow us to ?? potentially useful when the assumptions above break down???.

\section{Model and Definitions}\label{modeldefsec}
In this paper, we consider minimizing completion time in multiple-machine environment. There are $n$ jobs and a set of $m$ identical machines in the system. Each job $i$ is assigned a processing time $p_{i}$ and arrival time $r_{i}$, and we use $\Delta r_{i}=r_{i}-r_{i-1}$ to denote the interarrival time between job $i-1$ and $i$. We focus on \emph{work-conserving disciplines}, which is formally defined as following.
%never idles machines when there exists at least one feasible job or task awaiting the execution in the system. 
%We let $\Pi_{\mathcal{W}}$ denote the class of work-conserving disciplines.

\begin{definition}[Work-conserving scheduling discipline~\cite{harchol2013performance}] A scheduling discipline $\pi$ is called work-conserving if it never idles machines when there exists at least one feasible job or task awaiting the execution in the system. Here a job or task is called feasible, if it satisfies all the given constraints of the system (e.g, precedence constraint, preemptive and non-preemptive constraint, etc).
\end{definition}

Our result in this paper holds for both \emph{preemptive} and \emph{non-preemptive} models. In the preemptive model, the job that is running can be interrupted and later continued on any machine, while the system must follow the ``run to completion'' rule in the non-preemptive setting. 
For any scheduling discipline $\pi$, we compare it with an oblivious adversary, \ie, the optimal offline algorithm, for which there are no restrictions, it can have full knowledge of the input sequence in advance, together with the choices of $\pi$.

\begin{definition}[Competitive Ratio]
Let $\mathcal{CR}_{\pi}$ denote the competitive ratio of discipline $\pi$. It is defined as 
\begin{align*}
\mathcal{CR}_{\pi}=\max_{I}\frac{G_{\pi}(I)}{G_{\pi^{*}}(I)},
\end{align*}
where we use $G_{\pi}(I)$ to denote the objective value under discipline $\pi$ and instance $I$.
\end{definition}
%Formally, discipline $\pi$ is asymptotic optimal if 
%\begin{align}
%\mathbbm{P}_{I}\Big(\lim_{n\rightarrow \infty}\frac{G^{(n)}_{\pi}(I)}{G^{(n)}_{\pi^{*}}(I)}=1\Big)=1 \;(\forall \pi \in \Pi_{\mathcal{W}}),
%\end{align}
%where we use $G^{(n)}_{\pi}$ to represent the total completion time under discipline $\pi$, when there are $n$ jobs in the system.

%It is worth pointing out that in our definition, the optimality condition holds for almost every sample path, which is different from that in~\cite{} (Chou et al.), which requires that $\lim_{n\rightarrow \infty}\frac{\mathbbm{E}_{I}(G^{(n)}_{\pi}(I))}{\mathbbm{E}_{I}(G^{(n)}_{*}(I))}=1$.
 
%Besides the academic interests on optimizing completion time, we ???
\subsection{Remark on the assumptions}
%\paragraph{Assumptions.} 
In this paper we use $\mu^{(k)}_{p},\mu^{(k)}_{r}$ to denote the expected value of the $k$-th job workload and interarrival time respectively. We note that the stability condition, \ie, Assumption~\ref{stableassump}, can be replaced by the following assumption.

\begin{assumption}\label{assumptionrho}
$\rho^{(n)}\leq 1+o(n^{-1/2})$, where $\rho^{(n)}$ is defined as
\begin{align*}
\rho^{(n)}=\sup_{\ell\in [n]}\frac{\mu^{(\ell)}_{p}}{\sum_{k=\ell}^{\ell+m-1}{\mu^{(k)}_{r}}}.
\end{align*}
\end{assumption}
In addition, the mean values of interarrival time are not necessary to be identical, as long as the following assumption holds.
\begin{assumption}\label{assumptionstable}
For the mean values of interarrival time, it suffices to have one of the following conditions,
\begin{itemize}
\item The mean values of interarrival time are almost identical, \ie,
$|\mu^{(i)}_{r}-\mu^{(j)}_{r}|=o(n^{-1/2})\; (\forall i,j\in [n])$.
\item $\{\mu^{(k)}_{r}\}$ is non-decreasing, \ie, $\mu^{(k)}_{r}\leq \mu^{(k+1)}_{r}$.
\end{itemize}
A special case is when the mean values of interarrival time are identical, \ie, there exists $\mu_{r}$ such that $\mu^{(k)}_{r}=\mu_{r}$ for $\forall k\in [n]$.
\end{assumption}

\section{Tight competitive ratio bound on flow time}\label{tightflowboundsec}
%Let $B=p_{\max}/p_{\min}$ be the ratio of the maximum job processing time and minimum job processing time, 
We establish a tight characterization on the performance of work-conserving disciplines on flow time in this section.  
\subsection{Upper bound}
\begin{theorem}\label{wclowbound}
The competitive ratio of any work-conserving scheduling discipline is no more than $2B$.
\end{theorem}
\begin{proof}
In the following proof, we use $W_{\pi}(t)$ to represent the remaining workload under discipline $\pi$ at time $t$, and let $\pi^{*}$ denote the optimal scheduling discipline. The main idea of our proof is to relate $n_{\pi}(t)$, the number of jobs alive under $\pi$ to that under the optimal discipline $\pi^{*}$, which is achieved by comparing the amount of unfinished workload under these two disciplines. The bounded job size ratio parameter allows us to convert the relation between remaining workload to that between the number of unfinished jobs.

To start with, observe that
\begin{align}
n_{\pi}(t)\leq \frac{W_{\pi}(t)}{p_{\min}}=&B\cdot \frac{W_{\pi^{*}}(t)+\Big(W_{\pi}(t)- W_{\pi^{*}}(t)\Big)}{p_{\max}}\tag{definition of $B$ and $p_{\min}$}\\
\leq & B \cdot \Big(n_{\pi^{*}}(t)+ \frac{W_{\pi}(t)- W_{\pi^{*}}(t)}{p_{\max}} \Big),	  (\forall t\geq 0) \label{aliverela}
\end{align}
where the last inequality is due to the fact that $n_{\pi^{*}}(t)\geq \frac{W_{\pi^{*}}(t)}{p_{\max}}$. In the following we let
\begin{align*}
\mathfrak{T}=\Big\{t\; \Big|\; n_{\pi}(t)+n_{\pi^{*}}(t)>0\Big\}
\end{align*}
be the set of non-trivial time slots, \ie, in which unfinished jobs exist either under $\pi$ or $\pi^{*}$, and let $\mathfrak{T}^{\sharp}_{\pi}$ denote the collection of time slots in which idle machines exist under discipline $\pi$. We argue that $W_{\pi}(t)$ is no more than $(m-1)\cdot p_{\max}$ for $\forall t\in \mathfrak{T}^{\sharp}_{\pi}$. To show this fact, note that $\pi$ is a work-conserving algorithm, which implies that there are less than $m$ unfinished jobs at time $t$ due to the existence of idle machines. As a consequence, we have
\begin{align}\label{workloaddifff}
\frac{W_{\pi}(t)- W_{\pi^{*}}(t)}{p_{\max}}\leq m-1,\;(\forall t\in \mathfrak{T}^{\sharp}_{\pi})
\end{align}
which follows from the non-negativity of $W_{\pi^{*}}(t)$. On the other hand, we claim that bound (\ref{workloaddifff}) still holds for $t\in \mathfrak{T}\setminus \mathfrak{T}^{\sharp}_{\pi}$, \ie, when all the machines are busy under $\pi$. To see this fact, for each $t\in \mathfrak{T}\setminus\mathfrak{T}^{\sharp}_{\pi}$, we define its related time slot in $\mathfrak{T}^{\sharp}_{\pi}$ as,
\begin{align*}
t^{\sharp}=\max\Big\{\bar{t}\;\Big|\; \bar{t} \in  \mathfrak{T}^{\sharp}_{\pi} \cap [0,t]\Big\},	
\end{align*}
we next claim that for $\forall t\geq 0$,
\begin{align}\label{remainingworkloadrela}
W_{\pi}(t)- W_{\pi^{*}}(t)\leq W_{\pi}(t^{\sharp})- W_{\pi^{*}}(t^{\sharp})\leq (m-1)\cdot p_{\max}.
\end{align}
This is because the remaining workload under $\pi$ deceases at the maximum speed of $m$ during time interval $(t^{\sharp},t]$, hence the difference of the remaining workload between $\pi$ and $\pi^{*}$ must be non-increasing in $(t^{\sharp},t]$. Combining (\ref{aliverela}) and (\ref{remainingworkloadrela}), we have
\begin{align}\label{numberofjobsbound}
n_{\pi}(t)\leq B\cdot \Big(n_{\pi^{*}}(t)+m-1 \Big) \; (\forall t\geq 0).	
\end{align}
Now we are ready to bound the total flow time of $\pi$ as following,
\begin{align}
F_{\pi}=\int_{t\in \mathfrak{T}}{n_{\pi}(t) dt} & =\int_{t\in \mathfrak{T}^{\sharp}_{\pi}}{n_{\pi}(t) dt}+\int_{t\in \mathfrak{T}\setminus\mathfrak{T}^{\sharp}_{\pi}}{n_{\pi}(t) dt}\notag\\
& \overset{(a)}{\leq} \int_{t\in \mathfrak{T}^{\sharp}_{\pi}}{n_{\pi}(t) dt}+\int_{t \in \mathfrak{T}\setminus\mathfrak{T}^{\sharp}_{\pi}}{ B \cdot \Big(n_{\pi^{*}}(t)+ m-1 \Big) dt}\notag \\
& \overset{(b)}{\leq} B\cdot \Big[\Big(\int_{t\in \mathfrak{T}^{\sharp}_{\pi}}{n_{\pi}(t) dt}+ \int_{t \in \mathfrak{T}\setminus\mathfrak{T}^{\sharp}_{\pi}}{(m-1) dt} \Big) + F_{\pi^{*}} \Big]\notag\\
& \overset{(c)}{\leq} 2B\cdot F_{\pi^{*}},\label{competitivework}
\end{align}
where $(a)$ follows from inequality (\ref{numberofjobsbound}); $(b)$ is based on the fact that $F_{\pi^{*}}=\int_{t \in \mathfrak{T}} n_{\pi^{*}}(t) dt\geq \int_{t \in \mathfrak{T}\setminus\mathfrak{T}^{\sharp}_{\pi}} n_{\pi^{*}}(t) dt$; Because $n_{\pi}(t)\geq m$ when $t\in \mathfrak{T}\setminus \mathfrak{T}^{\sharp}_{\pi}$, it follows that $\int_{t\in \mathfrak{T}^{\sharp}_{\pi}}{n_{\pi}(t) dt}+ m\cdot \int_{t \in \mathfrak{T}\setminus\mathfrak{T}^{\sharp}_{\pi}}{ dt}$ is a lower bound of $F_{\pi^{*}}$, hence $(c)$ holds. The proof is complete.
\end{proof}

%It is worth mentioning that we our performance bound is better than that in~\cite{zheng2013new}.

In the following proposition, we will see that our performance upper bound is indeed tight, up to a constant gap that is no greater than $4$.
\begin{proposition}\label{tightexample}
There exists a work-conserving scheduling algorithm $\pi$ with competitive ratio $\mathcal{CR}_{\pi}\geq B/2$.	
\end{proposition}

\begin{proof}
See Appendix~\ref{appendtightexample}.
\end{proof}

%Proposition~\ref{tightexample} indicates that the competitive ratio upper bound for work-conserving scheduling algorithms in Lemma~\ref{wclowbound} is indeed tight (up to a constant that is no more than $?$).
Now we are ready to establish the competitive ratio supremum on the class of work-conserving disciplines.
\begin{theorem}\label{tightuppinfworkcon}
The supremum of competitive ratios of work-conserving scheduling disciplines satisfies
\begin{align}
\sup_{\pi\in \Pi_{\mathcal{W}}}\mathcal{CR}_{\pi}=\Theta(B).
\end{align}
Specifically, we have $B/2 \leq \sup_{\pi\in \Pi_{\mathcal{W}}}\mathcal{CR}_{\pi} \leq 2B$.	
\end{theorem}
\begin{proof}
Combining Theorem~\ref{wclowbound} and Proposition~\ref{tightexample}, the proof is complete.	
\end{proof}

\paragraph{Remark.} In the above, for a clean presentation, we present our results in the context of the fundamental model of $Pm|pmtn, r_{j}|\sum{F_{j}}$ and $Pm|r_{j}|\sum{F_{j}}$. However, it is worth pointing out that our upper bound and its tightness holds under several more general conditions, including examples showing below. We first note that under a more general setting, the correctness of Proposition~\ref{tightexample} can be easily verified, as the worst gap between LRPT and the optimal algorithm is non-decreasing with respect to the input instance set. Hence it remains to discuss the correctness of Theorem~\ref{wclowbound} for the following scenarios.
%in a larger input instance set, is no less than the ratio of $B/2$ shown in Proposition~\ref{tightexample}.
\begin{itemize}
\item \emph{Multitask job with arbitrary intertask precedence constraint.} It can be seen that Theorem~\ref{tightuppinfworkcon} still holds under arbitrary precedence constraint between tasks within the same job. This is because that precedence constraint on tasks will not change the fact that the number of jobs alive at $t\in \mathfrak{T}^{\sharp}_{\pi}$ is less than $m$, which is shown in equation~(\ref{workloaddifff}).
%In addition, we can let LRPT process tasks

\item \emph{Constant number of interjob precedence phase~\cite[Section 4.6]{gittins2011multi}.} With the appearance of interjob precedence, the number of jobs alive is not necessary equal to the the number of feasible jobs. Hence the existence of idle machines does not imply a lower bound of $m-1$ on the number of jobs alive, as jobs may be waiting for service due to the precedence constraint. However, if there are constant number of precedence phases, we are able to conclude that the $W_{\pi}(t)=O((m-1)\cdot p_{\max})$. As a consequence, the RHS of (\ref{numberofjobsbound}) and (\ref{competitivework}) are only blowed up by a constant and the competitive ratio upper bound is still in the order of $O(B)$. 	
\end{itemize}

\subsection{Lower bound}
In this section, we investigate the competitive ratio infimum of work-conserving algorithms. Firstly, for preemptive case,
%\paragraph{Known competitive ratio infimum for preemptive case.} 
it has been proved that the \emph{shortest remaining processing time} (SRPT) discipline, which always serves the job with shortest remaining processing time, minimizes the number of jobs in the system in single machine model. While for multiple machines, it is impossible to achieve a competitive ratio of $o(\log B)$~\cite{leonardi1997approximating}, which holds for both work-conserving and non-work-conserving disciplines. For non-preemptive case, similar as the proof in~\cite{leonardi1997approximating}, a lower bound of $B^{1-\varepsilon}$ on the competitive ratios of work-conserving disciplines can be shown, via reduction to the \emph{numerical three-dimensional matching} (N3DM) problem.

To summarize, we have the following lemma.
%The infimum result above was not explicitly stated in \cite{smith1978new,}, but can be easily seen according to basic definitions. However, no results were known for the non-preemptive work-conserving algorithms. In the following subsection, we establish a lower bound on the competitive ratios of work-conserving scheduling algorithms. Surprisingly, it almost matches the $\Theta(B)$ upper bound proven in last section.

\begin{lemma}\label{nonpreelow}
Let $\mathcal{W}_{p}$ and $\mathcal{W}_{\bar{p}}$ be the collection of preemptive and non-preemptive work-conserving disciplines. The infimum of preemptive work-conserving disciplines satisfies that
\begin{equation*}
\inf_{\pi\in \mathcal{W}_{p}}{\mathcal{CR}_{\pi}}=
\begin{cases}
1 & \text{$m=1$}\\
\log B& \text{$m\geq 2$}
\end{cases}
\end{equation*}
The competitive ratio infimum of work-conserving scheduling algorithms satisfies that $\inf_{\pi\in \Pi_{\mathcal{W}_{\bar{p}}}}\mathcal{CR}_{\pi}\geq B^{1-\varepsilon}$ for arbitrary positive number $\varepsilon>0$, , unless P $=$ NP. Consequently we have $\mathcal{CR}_{\pi}\in [B^{1-\varepsilon}, 2B]$ for any $\pi\in \mathcal{W}_{\bar{p}}$.
%Here we remark that the positive number $\varepsilon$ can be arbitrary small and not necessary to be a constant.	
\end{lemma}
\paragraph{Discussions.} For the more restricted class of non-size based scheduling algorithms, the following competitive ratio lower bound has been established in~\cite{DBLP:conf/soda/MotwaniPT93}. Combined with Theorem~\ref{wclowbound}, we can further obtain Theorem~\ref{nonsizeopt}.
\begin{fact}[\cite{DBLP:conf/soda/MotwaniPT93}]
No (deterministic) non-size based scheduling algorithm can achieve a competitive ratio that is less then $B$.  	
\end{fact}

\begin{theorem}\label{nonsizeopt}
All non-size based and work-conserving scheduling algorithms achieve the same competitive ratio (up to a constant of two). More specifically,
$\mathcal{CR}_{\pi}\in [B, 2B]$ holds for $\forall \pi \in \Pi_{\overline{s}}$. 	
\end{theorem}
Our conclusion can be regarded as the worst case counterpart of the well-known result for $M/G/1$ queue, which is summarized in the following Fact~\ref{mgonelemma}.
\begin{fact}
\label{mgonelemma}
For the case of an $M/G/1$ queue, all \emph{non-preemptive} and \emph{work-conserving} service orders that \emph{do not make use of job sizes} (\ie, \emph{non-size based}) have the same distribution of the number of jobs in the system~\cite{conway2003theory,harchol2013performance}.	
\end{fact}
Compared with Fact~\ref{mgonelemma}, it is also important to point out that Theorem~\ref{nonsizeopt} holds in the more general setting of multiple servers and arbitrary arrival distribution. In addition, our result indeed indicates that in the worst case, both preemptive and non-preemptive non-size-based scheduling algorithms achieve almost identical competitive ratio, while Fact~\ref{mgonelemma} only applies for non-preemptive algorithms.

\section{Optimality in stochastic online completion time minimization}\label{completiontime}
%
%[constant efficiency ratio implies asymptotic optimality under the metric of completion time]

In this section, we show the asymptotic optimality condition, utilizing our tight characterization on the worst case performance of work conserving algorithms with respect to the metric of flow time. 
%We show that any work conserving algorithm achieves asymptotic optimal total completion time, under mild probabilistic assumptions on the arrival process and workload distribution.
We first state the following fact that will be useful for establishing Theorem~\ref{completiontheo}.

%The Asymptotic Theory of Extreme Order Statistics.
\begin{lemma}[\cite{galambos1978asymptotic}]\label{randommax} For random variable sequence $\{X_{i}\}_{i\in [n]}$, the equation
\begin{align}\label{convergenceofmax}
\lim_{n\rightarrow \infty}\frac{\max_{i\in [n]}X_{i}}{n^{1/r}}=0	
\end{align}
holds almost surely, under one of the following conditions:
\begin{itemize}
\item $\{X_{i}\}_{i\in [n]}$ are \rm{i.i.d} and $\mathbbm{E}[X_{i}^{r}]<\infty$, \ie, the $r$-th moment of $X_{i}$ is finite.
\item $\mathbbm{E}[X_{i}^{r+\epsilon}]<\infty$ holds for some $\epsilon>0$, \ie, the $(r+\epsilon)$-th moment of $X_{i}$ is finite.
\end{itemize}
%Otherwise $\mathbbm{P}(\frac{\max_{i\in [n]}{X_{i}}}{n^{1/r}}=\infty)=1$.
In addition, $\mathbbm{P}(\frac{\max_{i\in [n]}{X_{i}}}{n^{1/r}}=\infty)=1$ when $\{X_{i}\}_{i\in [n]}$ are \rm{i.i.d} and $\mathbbm{E}[X^{r}]=\infty$.
% infinite $r$-th moment of an \rm{i.i.d} sequence implies an infinite limitation in (\ref{convergenceofmax}), \ie,
%when the number of random variables $n\geq N_{\varepsilon,\delta}= \Big(\frac{\mathbbm{E}[X^{r+\epsilon}]}{\delta \cdot \varepsilon^{r+\epsilon}} \Big)^{\frac{r}{\epsilon}}$, we have
%\begin{align}\label{convergencerateofmax}
%\mathbbm{P}\Big({\frac{\max_{i\in [n]}{X_{i}}}{{n^{1/r}}} \leq \varepsilon}\Big)\geq 1-\delta,
%\end{align}
%if the $(r+\epsilon)$-th moment of $X$ is finite.
\end{lemma}

\begin{proof}
The proof of the i.i.d distributed case mainly relies on the Borel-Cantelli Lemma, and proof of the general case simply utilizes the Markov inequality. The detailed proof is deferred to Appendix~\ref{appendixmaxlemma}.
\end{proof}

We next make the following observation.
\begin{observation}\label{observationofflowtime}
The asymptotic optimal condition holds for $\forall \pi\in \Pi_{\mathcal{W}}$, if the total flow time under the optimal scheduling algorithm satisfies that $\sum_{i\in [n]}{f^{*}_i}=o(n^2/B^{(n)})$, where $B^{(n)}=\frac{\max_{i\in [n]}p_{i}}{\min_{i\in [n]}p_{i}}$ represents the job size ratio.
\end{observation}
\begin{proof}
See Appendix~\ref{appendixob}.
\end{proof}

\paragraph{Remark.} We derive an $\Omega(n^{2})$ lower bound on optimal total completion time in the proof of Observation~\ref{observationofflowtime}, which still holds without any assumptions on the arrival process, if there is a lower bound $\Delta=\Theta(1)$ on job workload. To see this fact, we re-index the jobs by their completion time order as $C_{\sigma_{k}}\leq C_{\sigma_{k+1}}\;(\forall k\in [n-1])$, then we have $C_{\sigma_{k}}\geq \frac{\sum_{\ell \in [k]}{p_{\sigma_{\ell}}}}{m}$ and
\begin{align*}
\sum_{k\in [n]}{C_{k}}\geq \frac{1}{m}\sum_{k\in [n]}{\sum_{\ell \in [k]}{p_{\sigma_{\ell}}}}\geq \frac{(n+1)n}{2m}\cdot \Delta=\Omega(n^{2}).
\end{align*}

\subsection{Lower bound on the optimal flow time}

\begin{lemma}\label{singleopt}
For a single server system with interarrival time $\{\Delta r_{k}\}_{k\in [n]}$ and job workload  $\{p_{k}\}_{k\in [n]}$,
\begin{align*}
\mathbbm{P}\Big(\lim_{n\rightarrow \infty} \frac{\sum_{k\in [n]}{f^{*}_{i}}}{n^{3/2+\epsilon}} =0\Big)=1,
\end{align*}
for any $\epsilon>0$, if
\begin{align}\label{singlecon}
\frac{\sum_{k\in [n]}{\mu^{(k)}_{v}\cdot \mathbbm{1}_{\mu^{(k)}_{v}>0}}}{n} =o(n^{-1/2+\epsilon}).
\end{align}
\end{lemma}
\begin{proof}
In this proof, we let $W_{k}$ denotes the waiting time of the $k$-th arriving job under \emph{first come first serve} (FCFS) discipline. For general input job workload and arrival time distributions, we have the following recursion according to Lindley equation~\cite{asmussen2008applied},
\begin{align*}
W_{k+1}=(W_{k}+p_{k}-\Delta r_{k})^{+}\equiv (W_{k}+v_{k})^{+},
\end{align*}
where we let $v_{k}= p_{k}-\Delta r_{k}$, $\mu^{(k)}_{v}=\mathbbm{E}[v_{k}]=\mu^{(k)}_{p}-\mu^{(k)}_{r}$ for $\forall k\in [n]$ and $x^{+}=\max\{x, 0\}$ for $\forall x\in \mathbbm{R}$. Solving the recursive equation, it can be shown that~\cite{asmussen2008applied},
\begin{align*}
W_{n}=\max\Big\{T_{n}, T_{n}-T_{1},T_{n}-T_{2},\ldots,T_{n}-T_{n-1}, 0\Big\}=T_{n}-\min_{k\in [n]}{T_{k}},
\end{align*}
where $T_{k}=\sum_{i\in [k]}{v_{i}} \; (\forall k\in [n])$. Hence the total waiting time under FCFS is,
\begin{align*}
\sum_{k\in [n]}{W_{k}}&=\sum_{k\in [n]}{\Big[T_{k}-\min_{i\in [k]}{T_{i}}\Big]},\\
&=\underbrace{\sum_{k\in [n]}{\Big[T_{k}-\sum_{i\in [k]}{\mu^{(i)}_{v}} \Big]}}_{\Sigma_{1}}+ \underbrace{\sum_{k\in [n]}\Big[ \max_{i\in [k]}{\{-T_{i}\}}-\sum_{i\in [k]}{(-\mu^{(i)}_{v})}\Big] \Big\}}_{\Sigma_{2}},
%&= \Sigma_{1}+\Sigma_{2}.
\end{align*}

Note that
\begin{align*}
\frac{\Sigma_{1}}{n^{3/2+\epsilon}}&=\frac{ \sum_{k\in [n]} \Big\{T_{k}-\sum_{i\in [k]}{\mu^{(i)}_{v}} \Big\} }{n^{3/2+\epsilon}}\\
&=\sum_{k\in [n]}{\Big(\frac{1}{n^{1/2+\epsilon}}-\frac{k-1}{n^{3/2+\epsilon}}\Big) (v_{k}-\mu^{(k)}_{v})}.
%&=\frac{\sum_{k\in [n]}{ (v_{k}-\mu^{(k)}_{v})}}{n^{1/2}}-
%\sum_{k\in [n]}{\Big(\frac{k-1}{n^{3/2}}\Big)(v_{k}-\mu^{(k)}_{v})}.
\end{align*}
Applying the Chebyshev inequality, we know that for any $\varepsilon>0$,
\begin{align}\label{sigma1ineq}
\lim_{n\rightarrow \infty} \mathbbm{P}\Big( \frac{\Sigma_{1}}{n^{3/2+\epsilon}}\geq \varepsilon \Big) \leq &\lim_{n\rightarrow \infty} \frac{\var \Big( \sum_{k\in [n]}{\Big(\frac{1}{n^{1/2+\epsilon}}-\frac{k-1}{n^{3/2+\epsilon}}\Big) (v_{k}-\mu^{(k)}_{v})} \Big)}{\varepsilon^{2}}\notag\\
\leq &\lim_{n \rightarrow \infty} {\frac{\sum_{k\in [n]}{\var(v_{k}-\mu^{(k)}_{v})} }{n^{1+2\epsilon}\varepsilon^{2}}}\notag \\
=&\lim_{n \rightarrow \infty} { \frac{\sum_{k\in [n]} {({\sigma_{r}^{(k)}}^{2}}+{\sigma_{p}^{(k)}}^{2}) }{n^{1+2\epsilon}\varepsilon^{2}}}=0,
\end{align}
where the last inequality holds as $\sup_{n}\{{\sigma^{(n)}_{r}}^{2}, {\sigma^{(n)}_{p}}^{2} \}<\infty$. Hence $\lim_{n\rightarrow \infty}{\frac{\Sigma_{1}}{n^{3/2}}}= 0$ almost surely.

On the other hand, we have the following inequality for $\lambda=\Theta(n^{1/2+\epsilon})$,
\begin{align*}
\mathbbm{P}\Big( \Sigma_{2}\leq \lambda n  \Big)&= \mathbbm{P}\Big( \sum_{k\in [n]}{\Big[ \sum_{i\in [k]}{\mu^{(i)}_{v}}-\min_{i\in [k]}{T_{i}}\Big]}  \leq \lambda n\Big)\\
& \geq \mathbbm{P}\Big(\sum_{i\in [k]}{\mu^{(i)}_{v}}-\min_{i\in [k]}{T_{i}} \leq \lambda, \forall k\in [n] \Big)\\
& = \mathbbm{P}\Big(T_{i}\geq \sum_{j\in [k]}{\mu^{(j)}_{v}}-\lambda, \forall i\in [k], k\in [n] \Big)\\
&= \mathbbm{P}\Big( T_{i}\geq \sum_{j\in [k]}{\mu^{(j)}_{v}\cdot \mathbbm{1}_{\mu^{(j)}_{v}<0}}-\lambda^{\prime}_{k}, \forall i\in [k], k\in [n]   \Big),
\end{align*}
where $\lambda^{\prime}_{k}=\lambda-\sum_{j\in [k]}{\mu^{(j)}_{v}\cdot \mathbbm{1}_{\mu^{(j)}_{v}>0}}\; (\forall k\in [n])$. Combining with the facts that 
\begin{align*}
\lambda^{\prime}_{k}\leq \lambda^{\prime}=\lambda-\sum_{j\in [n]}{\mu^{(j)}_{v}\cdot \mathbbm{1}_{\mu^{(j)}_{v}>0}}
\end{align*}
for $\forall k\in [n]$, and $\sum_{i\in [k]}{\mu^{(i)}_{v}\cdot \mathbbm{1}_{\mu^{(i)}_{v}<0}}$ is non-increasing with respect to the index $k$, we further have
\begin{align}
%& \geq \mathbbm{P}\Big(  -T_{i}\leq \lambda^{\prime}+\sum_{j\in [i]}{(-\mu^{(j)}_{v} \cdot \mathbbm{1}_{\mu^{(i)}_{v}<0})} , \forall i\in [k], k\in [n] \Big)\\
\mathbbm{P}\Big( \Sigma_{2}\leq \lambda n  \Big)& \geq\mathbbm{P}\Big( T_{k}\geq \sum_{i\in [k]}{\mu^{(i)}_{v}\cdot \mathbbm{1}_{\mu^{(i)}_{v}<0}}-\lambda^{\prime}, \forall  k\in [n] \Big)\notag\\
&= \mathbbm{P}\Big(\min_{k\in [n]}\Big\{T_{k}-\sum_{i\in [k]}{\mu^{(i)}_{v}\cdot \mathbbm{1}_{\mu^{(i)}_{v}<0}} \Big\}+\lambda^{\prime} \geq 0 \Big)\notag\\
&\geq \mathbbm{P}\Big(\min_{k\in [n]}\Big\{T_{k}-\sum_{i\in [k]}{\mu^{(i)}_{v}} \Big\} +\lambda^{\prime} \geq 0\Big)\notag\\
& \geq \mathbbm{P} \Big( \max_{k\in [n]}\Big | T_{k}-\sum_{i\in [k]}{\mu^{(i)}_{v}}  \Big| \leq \lambda^{\prime} \Big)\notag\\
& \geq 1-\frac{ \sum_{k\in [n]} {({\sigma_{r}^{(k)}}^{2}}+{\sigma_{p}^{(k)}}^{2}) }{{\lambda^{\prime}}^{2}}\rightarrow 0, \label{kolmo}
%\tag{Kolmogorov's inequality}
\end{align}
%where $\lambda^{\prime}_{k}=\lambda-\sum_{j\in [k]}{\mu^{(j)}_{v}\cdot \mathbbm{1}_{\mu^{(j)}_{v}>0}}$ and $\lambda^{\prime}=\lambda-\sum_{k\in [n]}{\mu^{(k)}_{v}\cdot \mathbbm{1}_{\mu^{(k)}_{v}>0}} $.
where the last inequality follows from Kolmogorov's inequality~\cite{Durrett:2010:PTE:1869916}. Hence $\frac{\Sigma_{2}}{n^{3/2}}$ also converges to $0$ almost surely. Combined with (\ref{sigma1ineq}), we have
\begin{align*}
\mathbbm{P}\Big( \lim_{n\rightarrow \infty}\frac{\sum_{k\in [n]}{W_{k}}}{n^{3/2+\epsilon}}=0 \Big)=1.
\end{align*}
The minimum total flow time is no more than that incurred by FCFS, \ie,
\begin{align*}
F^{*}\leq F_{\mathrm{FCFS}}=\sum_{k\in [n]}{(W_{k}+p_{k})}=o(n^{3/2+\epsilon}), w.p.1.
\end{align*}
The proof is complete.
%Note that $\Big\{\Big(\frac{n-k+1}{n^{3/2}}\Big)(v_{k}-\mu^{(k)}_{v})\Big\}$ is a sequence of independent random variables with $\mathbbm{E}[\Big(\frac{n-k+1}{n^{3/2}}\Big)(v_{k}-\mu^{(k)}_{v})]=0$ and finite total variance, \ie,
%\begin{align*}
%\lim_{n\rightarrow \infty}\sum_{k \in [n]}{\var \Big[\Big(\frac{n-k+1}{n^{3/2}}\Big) \Big(v_{k}-\mu^{(k)}_{v}\Big) \Big]}\leq  \lim_{n \rightarrow}\frac{\sum_{k\in [n]}{\var(v_{k}-\mu^{(k)}_{v})}}{n}\leq \lim_{n\rightarrow \infty}\frac{\sum_{k\in [n]}{(\sigma^{2}_{r}}+\sigma^{2}_{p})}{n}.
%\end{align*}

\end{proof}

\begin{proposition}\label{singlestablepro}
If the single server system is stable, \ie, $\mathbbm{P}(\lim_{n\rightarrow \infty}W_{n}<\infty)=1$, then
\begin{align}\label{singleserverstable}
\frac{\sum_{k\in [n]}{\mu^{(k)}_{v}\cdot \mathbbm{1}_{\mu^{(k)}_{v}>0}  } }{n} =o(n^{-1/2}).
\end{align}
\end{proposition}
\begin{proof}
See Appendix~\ref{appendixpfstab}.
\end{proof}

\begin{proposition}\label{multistable}
The stability of a multiple server system with job workload $\{p_{k}\}_{k\in [n]}$ and interarrival time $\{\Delta r_{k}\}_{k\in [n]}$ implies that
\begin{align}\label{multiplestab}
\frac{\sum_{k\in [n]}{(\mu^{(k)}_{p}-m\mu^{(k)}_{r})^{+}}}{n} =o(n^{-1/2+\epsilon}).
\end{align}

%\begin{align}\label{multiplestab}
%\frac{\sum_{k\in [n]}{\mu^{(k)}_{v}\cdot \mathbbm{1}_{\mu^{(k)}_{p}>s\mu^{(k)}_{r}}  }}{n} =o(n^{-1/2}).
%\end{align}

%\begin{align}\label{multiplestab}
%\frac{\sum_{k\in [n]}{\mu^{(k)}_{v}\cdot \mathbbm{1}_{\mu^{(k)}_{p}>m\cdot \mu_{r}}}}{n} =o(n^{-1/2}),
%\end{align}
%where $\mu_{r}=\mathbbm{E}[\Delta r_{k}], \forall k\in [n]$.
\end{proposition}
\begin{proof}
We consider a single server system $\Sigma^{*}$ with the same input distributions of interarrival and service time, while the server is $m$ times as fast as that in the multiple server system $\Sigma^{(m)}$. We remark that the remaining workload in this single server system is always no more than that in $\Sigma^{(m)}$, since the the server in $\Sigma^{(*)}$ always reduces the workload at the same rate as the case when all the $m$ servers in $\Sigma^{(m)}$ are busy, while the newly arriving jobs in these two systems are identical. Hence if the remaining workload in $\Sigma^{*}$ goes to infinity, then the remaining workload in $\Sigma^{(m)}$ is also unbounded,
%\ie, the stability of $\Sigma^{(m)}$ implies that single server system $\Sigma^{*}$ is stable, 
thus condition (\ref{multiplestab}) follows from Proposition~\ref{singlestablepro}.
\end{proof}
\begin{lemma}\label{flowtimebound}
For any input instance, the optimal flow time satisfies that
\begin{align}
F_{\pi^{*}}=\sum_{i\in [n]}{f^{*}_{i}}=o(n^{3/2+\epsilon}),
\end{align}
under Assumption~\ref{completiontimeassump}, together with  Assumption~\ref{assumptionrho} or \ref{assumptionstable}.
\end{lemma}

\begin{proof}
To show the analytic bound on the optimal flow time, in the following we reduce the problem to the case $m=1$, by utilizing the \emph{first come first serve} (FCFS) rule as the benchmark algorithm to obtain suitable upper bounds in the single server system. Specifically, consider the simple cyclic job allocation scheme, in which the $j$-th arriving job goes to server $\sigma_{j}\equiv j\;(\bmod \; m)$. Let $A_{i}$ denote the set of jobs allocated to server $i$ and $n_{i}$ denote the size of $A_{i}$, \ie, $n_{i}=|A_{i}|\in\{\lfloor\frac{n}{m}\rfloor, \lfloor\frac{n}{m}\rfloor+1\}=\Theta(n)$ and $\sum_{i\in [m]}{n_{i}}=n$. Then jobs arrives at server $i$ with interarrival time $\{\Delta r^{(i)}_{j}\}_{j \in [n_{i}]}$ and workload $\{p^{(i)}_{j}\}_{j \in [n_{i}]}$ given by
\begin{align}
 \; p^{(i)}_{j}=p_{(j-1)m+i}, \bar{p}^{(i)}_{j}&=\mathbbm{E}[p^{(i)}_{j}]=\mu^{((j-1)m+i)}_{p},\label{definitionofpdelta1}\\
 \Delta r^{(i)}_{j}=\sum_{s=0}^{m-1}{\Delta r_{(j-1)m+i+s}},  \bar{\Delta} r^{(i)}_{j}&=\mathbbm{E}[\Delta r^{(i)}_{j}]=\sum^{m-1}_{s=0}{\mu^{((j-1)m+i+s)}_{r}},\label{definitionofpdelta2}
\end{align}
where we denote $\Delta r_{k}=0$ for $k<0$.

We first show that for $\forall i\in [m]$,
\begin{align}\label{singlequecon}
\sum_{j\in A_{i}}{(\bar{p}^{(i)}_{j}-\bar{\Delta} r^{(i)}_{j})^{+}}=o(n^{1/2+\epsilon}_{i}),
\end{align}
%=\sum_{s=0}^{m-1}{\mathbbm{E}[\Delta r_{(j-1)m+i+s}]}=
%In addition, let
%\begin{align*}
%\bar{p}^{(i)}_{j}&=\mathbbm{E}[p^{(i)}_{j}]=\mu^{((j-1)m+i)}_{p}, \\
%\bar{\Delta} r^{(i)}_{j}&=\mathbbm{E}[\Delta r^{(i)}_{j}]=\sum^{m-1}_{s=0}{\mu^{((j-1)m+i+s)}_{r}},
%\end{align*}
%and $\bar{v}^{(i)}_{j}=\bar{p}^{(i)}_{j}-\bar{\Delta} r^{(i)}_{j}=\mu^{((j-1)m+i)}_{v}$.
for which it suffices to prove that $\sum_{k\in [n]}{(\mu^{(k)}_{p}-\sum_{i=k}^{k+m-1}{\mu^{(i)}_{r}})^{+}}=o(n^{1/2+\epsilon})=o(n^{1/2+\epsilon}_{i})$, as it can be seen that
\begin{align*}
\sum_{j\in A_{i}}{(\bar{p}^{(i)}_{j}-\bar{\Delta} r^{(i)}_{j})^{+}}\leq \sum_{k\in [n]}{\Big(\mu^{(k)}_{p}-\sum_{i=k}^{k+m-1}{\mu^{(i)}_{r}}\Big)^{+}}\end{align*}
holds for $\forall i\in [m] $ according to the definitions in (\ref{definitionofpdelta1})-(\ref{definitionofpdelta2}). Observe that under the cyclic allocation,
\begin{itemize}
\item $\rho^{(n)}=\sup\frac{\mu^{(n)}_{p}}{\sum_{k=n}^{n+m-1}{\mu^{(k)}_{r}}}\leq 1+o(n^{-1/2})$ implies that $(\mu^{(n)}_{p}-\sum_{i=n}^{n+m-1}{\mu^{(k)}_{r}})^{+}\leq o(n^{-1/2})\cdot \sum_{i=n}^{n+m-1}{\mu^{(k)}_{r}}=o(n^{-1/2})\cdot \sup\mu^{(n)}_{r}$, consequently we have $\sum_{k\in [n]}{(\mu^{(k)}_{p}-\sum_{i=k}^{k+m-1}{\mu^{(k)}_{r}})^{+}}=o(n^{1/2})\cdot \sup\mu^{(n)}_{r}=o(n^{1/2})$.
\item The stability of the multi-server system implies that $\sum_{k\in [n]}{(\mu^{(k)}_{p}-m\mu^{(k)}_{r})^{+}} =o(n^{1/2+\epsilon})$, using the elementary inequality $(x+y)^{+}\leq x^{+}+y^{+}$, we are able to obtain inequality,
\begin{align*}
\sum_{k\in [n]}\Big(\mu^{(k)}_{p}-\sum_{i=k}^{k+m-1}{\mu^{(k)}_{r}}\Big)^{+}\leq &\sum_{k\in [n]}(\mu^{(k)}_{p}-m\mu^{(k)}_{r})^{+}+\sum_{k\in [n]}\Big(m\mu^{(k)}_{r}-\sum_{i=k}^{k+m-1}{\mu^{(i)}_{r}}\Big)^{+}\\
\leq& \sum_{k\in [n]}(\mu^{(k)}_{p}-m\mu^{(k)}_{r})^{+}+o(n^{1/2+\epsilon})=o(n^{1/2+\epsilon}).
\end{align*}
under the condition that
\begin{itemize}
\item $|\mu^{(i)}_{r}-\mu^{(j)}_{r}|=o(n^{-1/2})$, since $\sum_{k\in [n]}(m\mu^{(k)}_{r}-\sum_{i=k}^{k+m-1}{\mu^{(k)}_{r}})^{+}\leq (m-1)n\cdot\sup_{i,j}|\mu^{(i)}_{r}-\mu^{(j)}_{r}|=o(n^{1/2})$
\item $\{\mu^{(k)}_{r}\}_{k\in [n]}$ is non-decreasing, which implies that $(m\mu^{(k)}_{r}-\sum_{i=k}^{k+m-1}{\mu^{(k)}_{r}})^{+}=0$.
\end{itemize}

%\begin{itemize}
%\item When $|\mu^{(i)}_{r}-\mu^{(j)}_{r}|=o(n^{-1/2})$, $(\mu^{(k)}_{p}-\sum_{i=k}^{k+m-1}{\mu^{(k)}_{r}})^{+}\leq (\mu^{(k)}_{p}-m\mu^{(k)}_{r})^{+}+(m\mu^{(k)}_{r}-\sum_{i=k}^{k+m-1}{\mu^{(k)}_{r}})^{+}$
%\item When $\{\mu^{(k)}_{r}\}_{k\in [n]}$ is non-decreasing, we have $(\mu^{(k)}_{p}-\sum_{i=k}^{k+m-1}{\mu^{(k)}_{r}})^{+}\leq (\mu^{(k)}_{p}-m\mu^{(k)}_{r})^{+}$
%\end{itemize}
\end{itemize}

%(\ref{singlecon}) also holds for each server, \ie,
%\begin{align*}
%\sum_{j\in A_{i}}{\bar{v}^{(i)}_{j}}\cdot \mathbbm{1}_{\bar{v}^{(i)}_{j}>0}=\sum_{j\in A_{i}}\mu^{((j-1)m+i)}_{v}\cdot \mathbbm{1}_{\mu^{((j-1)m+i)}_{v}>0}\leq\sum_{k\in [n]}{\mu^{(k)}_{v}\cdot \mathbbm{1}_{\mu^{(k)}_{v}>0}} =o(n^{-1/2})=o(n^{1/2}_{i}).
%\end{align*}
By (\ref{singlequecon}) and Lemma \ref{singleopt}, we have $\mathbbm{P}(F^{(i)}_{\pi^{*}}=o(n^{3/2}_{i}))=1$ for $\forall i\in [m]$, consequently we know that $F_{\pi^{*}}=\sum_{i\in [m]}{F^{(i)}_{\pi^{*}}}=\sum_{i\in [m]}o(n^{3/2}_{i})=o(n^{2})$ holds almost surely. The proof is complete.
%By the independence of the $m$ single server systems, we have $\mathbbm{P}(F^{(i)}_{\pi^{*}}=o(n^{3/2}_{i}), \forall i\in [m])=1$ and $F_{\pi^{*}}=\sum_{i\in [m]}{F^{(i)}_{\pi^{*}}}=\sum_{i\in [m]}o(n^{3/2}_{i})=o(n^{2})$ holds almost surely, according to the inequality $\sum_{i\in [m]}{n^{\alpha}_{i}}\leq (\sum_{i\in [m]}n_{i})^{\alpha}= n^{\alpha}\; (\forall \alpha\geq 1)$. The proof is complete.
\end{proof}

\subsection{Putting things together}

\begin{proofof}{Theorem~\ref{completiontheo}}

We are able to show the conclusion by Theorem~\ref{wclowbound}, Lemma~\ref{randommax} and Lemma~\ref{flowtimebound}, if there exists a constant lower bound on the minimum job workload. However, the job size could be arbitrary small. We finish the proof by considering the work-conserving algorithm that has the worst performance.

Note that for any given input $\{(p_{i},r_{i})\}_{i\in [n]}$, there are finite number of choices to schedule jobs in a work-conserving way. Indeed, all the jobs must be completed before time $t=r_{n}+\sum_{i\in [n]}{p_{i}}$ and there are at most $\binom{n}{m}$ different allocation schemes at each time slot. We remark that it suffices to prove the conclusion for algorithm $\mathcal{W}$, the ``worst'' work-conserving algorithm. For every given input instance, $\mathcal{W}$ always follows the work-conserving order that incurs  the largest total flow time. It is important to note that, there may be certain restrictions in specific models, for example, preemptive and non-preemptive constraints may exist. However, for the work-conserving orders considered in the definition of $\mathcal{W}$, the only requirement is to keep the machines busy whenever there are jobs alive. Hence $\mathcal{W}$ may not be a feasible solution to the problem considered, but can be always regarded as a universal bound on the set of feasible work-conserving algorithms, \ie, $F_{\pi}\leq F_{\mathcal{W}}$ always holds for any feasible work-conserving algorithm $\pi$.

For any fixed $\Delta$, we consider another benchmark system $\Sigma^{\prime}$, where the sizes of jobs with workload below the threshold $\Delta$ are increased to $\Delta$. Formally, the probability density function of job size in system $\Sigma^{\prime}$ is
\begin{align*}
f^{\prime}(p)=f(p)\cdot \mathbbm{1}_{p>\Delta}+\Delta \cdot \delta(p-\Delta),
\end{align*}
where $\delta(\cdot)$ represents the Dirac delta function and $f(\cdot)$ denotes the probability density function in the original system $\Sigma$. Then the total amount of flow time under algorithm $\mathcal{W}$ in the original system $\Sigma$ is no more than that in $\Sigma^{\prime}$. Since for every instance, the worst work-conserving order in system $\Sigma$ can be modified to a feasible work-conserving order in $\sigma^{\prime}$, by processing the additional workload at the end. It can be seen that this modified order is work-conserving and the incurred total flow time is no less than that incurred by $\mathcal{W}$ in $\Sigma$, and is no more than that incurred by $\mathcal{W}$ in $\Sigma^{\prime}$. Hence we can conclude that $F_{\mathcal{W}}\leq F^{\prime}_{\mathcal{W}}$.

Note that for the modified system $\Sigma^{\prime}$ and $\mathcal{W}$, 
\begin{align}
F^{\prime}_{\mathcal{W}}\leq {B^{(n)}}^{\prime}\cdot F^{\prime}_{\pi^{*}},
\end{align}
where ${B^{(n)}}^{\prime}=p_{\max}/\Delta$ satisfies that
\begin{align}
\lim_{n\rightarrow \infty} \frac{{{B}^{(n)}}^{\prime}}{n^{1/2-\epsilon}}= \lim_{n\rightarrow \infty} \frac{\max_{i\in [n]} p_{i}}{\Delta \cdot n^{1/2-\epsilon}}=0, w.p.1,
\end{align}
according to Lemma~\ref{randommax}. In addition, the $\alpha$-th moment of the job size in system $\Sigma^{\prime}$ is also finite,
\begin{align*}
\mathbbm{E}[{p^{\prime}}^{\alpha}]=\int_{[0,\Delta]}{{p^{\prime}}^{\alpha} f^{\prime}(p^{\prime}) dp^{\prime}}+\int_{(\Delta,\infty)}{{p^{\prime}}^{\alpha}f^{\prime}(p^{\prime})dp^{\prime}}\leq \Delta^{\alpha}+\mathbbm{E}[p^{\alpha}]<\infty, \forall \alpha>0,
\end{align*}
and the job workload distributions are indepedent. Hence we are able to conclude that $\mathbb{P}(F^{\prime}_{\mathcal{W}}=o(n^{3/2}))=1$, which implies that $\mathbb{P}(F_{\mathcal{W}}=o(n^{3/2}))=1$ and $\mathbb{P}(F_{\pi}=o(n^{3/2}))=1$ for any work-conserving algorithm $\pi$. Combining with Observation~\ref{observationofflowtime}, the proof for identical machine setting is complete.

\paragraph{Machines with different speed.} Now we prove the optimality condition for machines in parallel with different speeds. We would like to point out that our result still holds when there are $m$ machines in parallel with different speeds and the mean values of job interarrival time are identical. 

In the following we use $s_{i}$ to denote the speed of machine $i$. Firstly, to bound $F_{\pi^{*}}$, we reduce the problem to single server system by assigning each arriving job to machine $i$ with probability
\begin{align*}
p_{i}=\frac{s_{i}}{\sum_{j=1}^{m}{s_{j}}},  
\end{align*}
which is similar as~\cite{chen2007probabilistic}.
%$p_{i}=\frac{s_{i}}{\sum_{j=1}^{m}{s_{j}}}$ %$p_{i}=s_{i}/\sum_{j\in [m]}{s_{j}}$,
It can be seen that the job interarrival time at machine $i$ can be expressed as the summation of $n_{i}$ random variables with mean value of $\mu_{r}$, where $n_{i}$ follows Geometric distribution with success probability $p_{i}$. Then the mean value of interarrival time at machine $i$ is equal to 
\begin{align*}
{\bar{\Delta}r}^{(i)}_{j}=\mathbbm{E}[\Delta r^{(i)}_{j}]=\mu_{r}/p_{i}=\frac{\mu_{r}\cdot \sum_{j\in [m]}{s_{j}}}{s_{i}}.    
\end{align*}
Similar as the proof of Lemma~\ref{multistable}, stability of the multi-server system implies that 
\begin{align*}
\sum_{i\in [n]}{\Big(\frac{\mu^{(i)}_{p}}{\sum_{j\in [m]}{s_{j}}}-\mu_{r}\Big)^{+}}=o(n^{1/2}),
\end{align*}
from which we can obtain the $o(n^{3/2})$ bound on the optimal flow time at each machine, based on Lemma~\ref{singleopt} and the following fact,
\begin{align*}
\sum_{j\in A_{i}}{\Big(\frac{\mu^{(j)}_{p}}{s_{i}}-\bar{\Delta} r^{(i)}_{j}\Big)^{+}}=\sum_{j\in A_{i}}{\Big(\frac{\mu^{(j)}_{p}}{s_{i}}-\frac{\mu_{r}\cdot \sum_{j\in [m]}{s_{j}}}{s_{i}}\Big)^{+}}\leq \frac{\sum_{j=1}^{m}{s_{j}}}{s_{i}}\cdot \sum_{i\in [n]}{\Big(\frac{\mu^{(i)}_{p}}{\sum_{j\in [m]}{s_{j}}}-\mu_{r}\Big)^{+}}=o(n^{1/2}).
\end{align*}
The remaining proof is similar as the identical machine setting. For example, similar as Theorem~\ref{wclowbound}, we can show that
%$F_{\pi}\leq B\cdot [(\int_{t\in \mathfrak{T}^{\sharp}_{\pi}}{n_{\pi}(t) dt}+ \int_{t \in \mathfrak{T}\setminus\mathfrak{T}^{\sharp}_{\pi}}{(m-1) dt} ) + F_{\pi^{*}}]\leq (2+\frac{m-1}{\sum_{j\in [m]}{s_{j}}})B F_{\pi^{*}}$
\begin{align}
F_{\pi}& \leq B\cdot \Big[\Big(\int_{t\in \mathfrak{T}^{\sharp}_{\pi}}{n_{\pi}(t) dt}+ \int_{t \in \mathfrak{T}\setminus\mathfrak{T}^{\sharp}_{\pi}}{(m-1) dt} \Big) + F_{\pi^{*}} \Big]\notag\\
&\leq \Big(2+\frac{m-1}{\sum_{j\in [m]}{s_{j}}}\Big)B \cdot F_{\pi^{*}}=\Theta(B)\cdot F_{\pi^{*}}.
\end{align}
The proof is complete.
\end{proofof}

\section{Further Generalizations}\label{secgeneralization}
%\subsection{Generalization to weighted completion time minimization}

\paragraph{Minimizing weighted completion time.}
Indeed for the more general problem of minimizing weighted total completion time, the asymptotic optimality of work-conserving algorithms still hold, which can be proved via the same arguments as the unit weight case. 

\begin{proposition}\label{weightcompletiontheo}	
Any work conserving algorithm $\pi$ is almost surely asymptotically optimal for minimizing weighted total completion time under Assumption~\ref{completiontimeassump} or \ref{assumptionrho}, and the interarrival time, job workload and weight $\{\omega_{k}\}_{k\in [n]}$ defined on $[1,+\infty)$ are independent with
\begin{itemize}
\item finite second, $\alpha$-th and $\beta$-th moments respectively, where $1/\alpha+1/\beta=1-\epsilon$.
\item finite second, $(2+\epsilon)$-th moments and finite generating function (\ie, $\mathbbm{E}[e^{\epsilon \omega_{k}}]<\infty \;\forall k\in[n]$) respectively.
\end{itemize}

% variance assumption (\ref{completiontimeassump}) holds.
\end{proposition}

\begin{proof}
See Appendix~\ref{appendixpro}.
\end{proof}

%(\ref{totalarrivaltime})

%\paragraph{Remark.} We would like to emphasize that the competitive ratio upper bound in Theorem~\ref{wclowbound} is crucial to establish the asymptotic result based on the assumed probabilistic structure. Moreover,

%the condition $\sum_{i\in [n]}{f^{*}_i}=o(n^2)$ in Observation~\ref{observationofflowtime} on the minimum flow time is non-trivial without our upper bound on competitive ratio in Theorem~\ref{wclowbound}.

%We would like to emphasize that we only make an assumption on the ratio of job workload instead of the absolute value of workload, and the ratio is allowed to in ??. Even when the job size ratio is bounded by a constant, the total job workload, which is a straightforward lower bound on the total flow time, could be extremely large compared with the total arriving time.

%It is worth pointing out that $\sum_{i=1}^{n}{f^{*}_i}=o(n^2)$ is not a strong condition.

%Indeed we have $\max_{i}{C_{i}}=\Omega(n)$ since the total workload is at most $\sum_{i=1}^{n}{x_{i}}\leq n\cdot B =\Omega(n)$. And at each time slot, the number of unfinished jobs are at most $n$, thus $\sum_{i=1}^{n}{f_{i}}\leq nB\cdot n=O(n^2)$. Note that this $O(n^2)$ upper bound is very rough, and as long as we can figure out one ``a little bit good'' scheduling policy, under which the total flow time is in the order of $o(n^2)$, then we can reach the conclusion that the total completion time of any work-conserving scheduling policy is asymptotic optimal.

%Here we give an example in which .
%In the following we give an example

\paragraph{Relaxing the independence assumption.} It is clear to see that our analysis indeed carry over beyond the independence assumptions on job workload and arrival process. The asymptotic optimality condition requires nothing more than the convergence results in inequalities (\ref{totalarrivalpro}), (\ref{sigma1ineq}) and (\ref{kolmo}). We remark that Theorem~\ref{completiontheo} can be indeed generalized to the setting when Assumption \ref{completiontimeassump} is replaced by the following condition. The proof is deferred to Appendix~\ref{appendfact}.

\begin{assumption} There exists $\{u^{(k)}_{p}\}_{k\in [n]}$ and $\{u^{(k)}_{r}\}_{k\in [n]}$ such that for $\forall 1\leq i\leq j\leq n$,
\begin{align}
\mathbbm{E}\Big[\sum^{j}_{k=i}{(p_{k}-\mu^{(k)}_{p})}\Big]^{2}<\sum_{i\leq \ell\leq j}{u^{\ell}_{p}}, \;\mathbbm{E}\Big[\sum^{j}_{k=i}{(\Delta r_{k}-\mu^{(k)}_{r})}\Big]^{2}<\sum_{i\leq \ell\leq j}{u^{\ell}_{r}}.\label{gencon}
%\;(\forall 1\leq i\leq j\leq n,).
\end{align}
%\begin{align*}
%\mathbbm{E}[|S_{j}-S_{i}|^{\alpha}]\leq \sum_{i\leq \ell\leq j}{u_{\ell}}, 0\leq i\leq j\leq n.
%\end{align*}
\end{assumption}

%\begin{align*}
%\mathbbm{E}[|S_{j}-S_{i}|^{2}]=\sum^{j}_{k=i}\mathbbm{E}[X^{2}_{k}]+\sum_{i\leq \ell\neq k\leq j}\mathbbm{E}[X_{\ell}X_{k}]\leq \sum^{j}_{k=i}\mathbbm{E}[X^{2}_{k}],
%\end{align*}
%if $\{X_{k}\}_{k\in [n]}$ are negative correlated, \ie, $\cov(X_{\ell}, X_{k})=\mathbbm{E}[X_{\ell} X_{k}]-\mathbbm{E}[X_{\ell}]\mathbbm{E}\mathbbm[X_{k}]\leq 0$.

%For positively quadrant dependent sequence, the conclusion still holds as $\mathbbm{P}(a_{k}<M^{(n)}_{1},\forall k\in [n])\leq \prod_{k\in [n]}\mathbbm{P}(a_{k}<M^{(n)}_{1},\forall k\in [n])$. ()
%$\mathbbm{P}(a^{(n)}_{k}\leq M)=1-p_{k}$ and
%\begin{fact} The waiting time sequence $\{W_{k}\}_{k\in [n]}$ is positive quadrant dependent, \ie,
%\begin{align*}
%\mathbbm{P}(W_{i}\leq s, W_{j}\leq t)\leq \mathbbm{P}(W_{i}\leq s) \cdot \mathbbm{P}(W_{j}\leq t), \forall i,j\in[n], s,t\geq 0.
%\end{align*}
%Furthermore,
%\begin{align*}
%\mathbbm{P}(W_{i}\leq s_{i}, \forall i\in \mathcal{I})\leq \prod_{i\in\mathcal{I}}\mathbbm{P}(W_{i}\leq s_{i}), \forall \mathcal{I}\subseteq [n].
%\end{align*}
%\end{fact}
%\begin{proof}

%\end{proof}

\section{Numerical Results}\label{numericalsec}
In this section we conduct simulations to validate the convergence of competitive ratios of various work-conserving disciplines. We consider a computing system with $m=20$ machines. Job processing times are i.i.d distributed and follow from exponential distribution with mean $\mu$. Jobs arrive according to a Poisson process with rate $\lambda=m\rho\mu$, where $\rho$ represents the traffic intensity. As the system will be less congested when the traffic intensity $\rho$ is small, intuitively the resulting total completion time should be close to the minimum total completion time. Therefore we focus on scenarios when $\rho$ is close to $1$. More specifically, we let $\mu=1/40$, $\lambda=0.45$  and $0.49$, where $\rho=0.9$ and $\rho=0.98$ respectively.

As shown in Figure~\ref{figfcfs}--\ref{figrn}, for each discipline, we plot the \emph{empirical distribution function}, \ie, the estimated cumulative distribution function (CDF), of the ratio between total completion time incurred and that incurred under the optimal discipline $\mathrm{OPT}$. It is worth mentioning that finding the exact total completion under $\mathrm{OPT}$ is computationally expensive, due to the NP-hardness and exponential search space. We use the total arrival time, an explicit lower bound of total completion time, to calculate the ratios. Hence, the ratios in the figures are indeed pessimistic estimations of the real value and for every fixed sample path, the ratios under different disciplines are amplified by the same factor. The list of (work-conserving) scheduling disciplines considered and corresponding results are summarized as following.
\begin{itemize}
\item \emph{First Come First Serve} (FCFS). FCFS is the simplest form of scheduling algorithm, which always processes the jobs by the order of their arrival. FCFS is the default scheduler in Hadoop. Results are shown in Figure~\ref{figfcfs}.
\begin{figure}[H]
\centering
\vspace{-0.5cm}
 %\subfloat[$\rho=0.9$]{
 \subfloat{
 \includegraphics[scale=0.72]{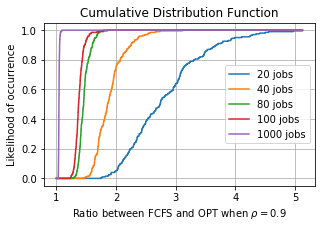}
 \label{fcfs1}
 }
 %\subfloat[$\rho=0.98$]{
 \subfloat{
 \includegraphics[scale=0.72]{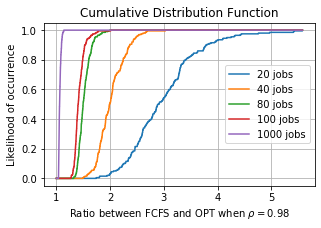}
 \label{fcfs2}
 }
 \vspace{-0.2cm}
 \caption{Empirical Distribution Function under FCFS}
 \vspace{-0.1cm}
\label{figfcfs}
\end{figure}
\vspace{-0.5cm}
\item \emph{Shortest Remaining Processing Time} (SRPT). As it is well-known, jobs with lower remaining processing time have a higher priority under SRPT. SRPT is efficient in optimizing the metric of mean response time and has been applied in several real life applications, including web servers~\cite{Harchol-Balter00implementationof}. Results are shown in Figure~\ref{figsrpt}.

% implementation of SRPT scheduling in web servers 
\begin{figure}[H]
\centering
\vspace{-0.5cm}
 %\subfloat[$\rho=0.9$]{
 \subfloat{
 \includegraphics[scale=0.72]{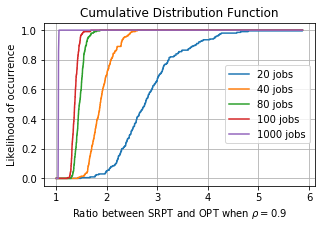}
 \label{srpt1}
 }
 %\subfloat[$\rho=0.98$]{
 \subfloat{
 \includegraphics[scale=0.72]{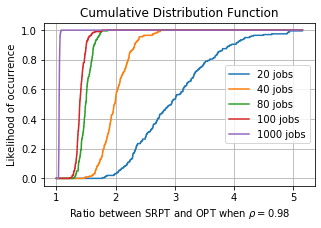}
 \label{srpt2}
 }
  \vspace{-0.2cm}
 \caption{Empirical Distribution Function under SRPT}
 \vspace{-0.1cm}
\label{figsrpt}
\end{figure}
\vspace{-0.5cm}

\item \emph{Longest Remaining Processing Time} (LRPT). As opposed to SRPT, LRPT always processes the job with longest remaining processing time. Since large jobs are handled slowly, we can get a sense of how poor that the performances of work-conserving algorithms can be by considering LRPT. Results are shown in Figure~\ref{figlrpt}.

\begin{figure}[H]
\centering
\vspace{-0.5cm}
 %\subfloat[$\rho=0.9$]{
 \subfloat{
 \includegraphics[scale=0.72]{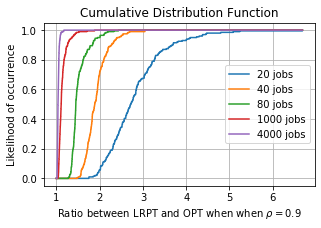}
 \label{lrpt1}
 }
 %\subfloat[$\rho=0.98$]{
 \subfloat{
 \includegraphics[scale=0.72]{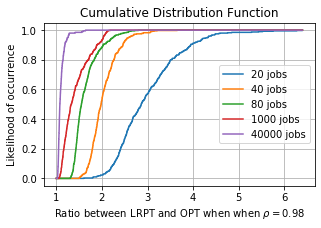}
 \label{lrpt2}
 }
  \vspace{-0.2cm}
 \caption{Empirical Distribution Function under LRPT}
 \vspace{-0.1cm}
\label{figlrpt}
\end{figure}
\vspace{-0.5cm}

\item \emph{Shortest Processing Time} (SPT). In this scheduling algorithm, jobs are executed according to order of the job processing time. SPT is known to be optimal for minimizing completion time in single machine setting. Results are shown in Figure~\ref{figspt}.

\begin{figure}[H]
\centering
\vspace{-0.5cm}
%\subfloat[$\rho=0.9$]{
\subfloat{
 \includegraphics[scale=0.72]{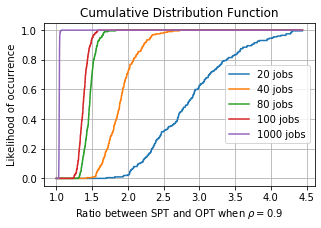}
 \label{spt1}
 }
 %\subfloat[$\rho=0.98$]{
 \subfloat{
 \includegraphics[scale=0.72]{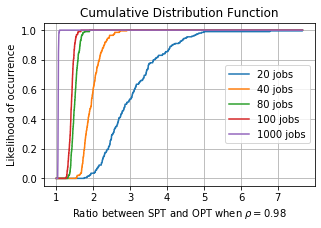}
 \label{spt2}
 }
 \vspace{-0.2cm}
 \caption{Empirical Distribution Function under SPT}
 \vspace{-0.1cm}
\label{figspt}
\end{figure}
\vspace{-0.5cm}
\item \emph{Random}. Different from the deterministic polices above, we randomly allocate the resources and each available job is selected with equal probability. By using randomness to make decisions, \emph{Random} is easy to implement. In addition, it requires almost no information about the system state and thus incurs very little overhead. Results are shown in Figure~\ref{figrn}.

\begin{figure}[H]
\centering
\vspace{-0.5cm}
 %\subfloat[$\rho=0.9$]{
 \subfloat{
 \includegraphics[scale=0.72]{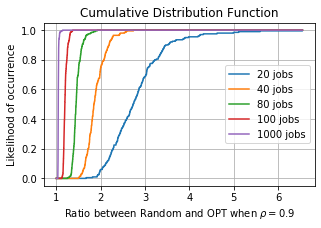}
 \label{fcfs1}
 }
 %\subfloat[$\rho=0.98$]{
 \subfloat{
 \includegraphics[scale=0.72]{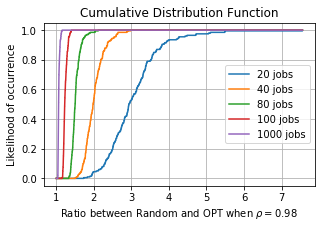}
 \label{fcfs2}
 }
 \vspace{-0.2cm}
 \caption{Empirical Distribution Function under Random}
 \vspace{-0.1cm}
\label{figrn}
\end{figure}
\end{itemize}
\vspace{-0.5cm}

\paragraph{Discussion on the numerical results.} We can see that for all the disciplines and sample paths in the experiments, the total completion time has a small constant gap (no more than $7$) to the optimum. This result coincides with the intuition that completion time is a relatively robust evaluation metric. On the other hand, almost for each fixed discipline, the ratio between the total completion time and the optimum total completion time incurred by a large number of jobs, is stochastically smaller than (first-order stochastically dominated\footnote{Random variable $A$ has a first order dominance over random variable $B$ if $\mathbbm{P}(A\geq x)\geq \mathbbm{P}(B\geq x)$ for all $x$ and  for some $x$, $\mathbbm{P}(A\geq x)> \mathbbm{P}(B\geq x)$.} by) that incurred by a smaller number of jobs. In addition, the empirical CDF converges to the unit step function at ratio of $1$, which verifies our asymptotic optimality conclusion. Note that in the results above, the gap between work-conserving disciplines and $\mathrm{OPT}$ is close to $1$ when the number of jobs is in the order of $10^{3}$, which is common in large scale applications. 
%\footnote{A cumulative distribution $F$ first-order stochastically dominates distribution $G$ iff $F(x)\leq G(x)$ for all $x$ with a strict inequality over some interval.} 

\section{Conclusion}\label{conclusionsec}
In this paper, we proved that in parallel machine environment, all work-conserving disciplines are asymptotic optimal in minimizing total completion time, as long as the interarrival time and job workload are independently distributed and have finite second and $(2+\varepsilon)$-th moment respectively, while the mean interarrival time are almost identical or non-decreasing. We further discussed simple generalization to weighted completion time minimization and showed possible relaxations on the independence assumption. To establish the result, we also obtained a complete characterization in competitive ratio bounds for work-conserving disciplines and the objective of flow time.

\bibliography{scheduling}
\bibliographystyle{plain}

\appendix

\section{Proof of Proposition~\ref{tightexample}}\label{appendtightexample}

\begin{proof}
It is known that SRPT achieves good performance guarantee with respect to flow time, hence it is natural to turn to the opposite direction, the \emph{longest remaining processing time first} (LRPT) discipline, when considering upper bound side of competitive ratio.

Consider a system consisting of $m$ large jobs with size $B$ and $mn$ small jobs with unit size. All the large jobs arrive at time $t=0$ and $m$ small jobs arrive at the beginning of every time slot $t\in [0,n-1]$. It is clear to see that the total flow time under LRPT is equal to
\begin{align*}
F_{\mathrm{LRPT}}=m\cdot B+ mn\cdot B=m(n+1)\cdot B,	
\end{align*}
while the total flow time under SRPT is
\begin{align*}
F_{\mathrm{SRPT}}= m\cdot (n+B)+mn=2mn+mB.	
\end{align*}
Hence the competitive ratio of LRPT is no less than
\begin{align*}
C_{\mathrm{LRPT}}=\frac{F_{\mathrm{LRPT}}}{F_{\pi^{*}}}\geq \frac{F_{\mathrm{LRPT}}}{F_{\mathrm{SRPT}}}\Big|_{n\rightarrow \infty}= \frac{m(n+1)\cdot B}{2mn+mB}\Big|_{n\rightarrow \infty}=\frac{B}{2}.	
\end{align*}
As LRPT is work-conserving, the proof is complete.
\end{proof}

\section{Proof of Lemma~\ref{randommax}}\label{appendixmaxlemma}

\begin{proof}
The proof of the i.i.d distributed case mainly relies on the Borel-Cantelli Lemma, and proof of the general case simply utilizes the Markov inequality.
\paragraph{Case $1$.} When $\{X_{i}\}_{i\in[n]}$ are \rm{i.i.d} distributed, we first consider the sequence $\{\frac{X_{i}}{i^{1/r}}\}_{i\in [n]}$ instead. Notice that for $\forall \varepsilon>0$,
\begin{align}
\sum_{i=1}^{\infty}{\mathbbm{P}}	\Big( \frac{X_{i}}{i^{1/r}}\geq \varepsilon \Big)=\sum_{i=1}^{\infty}{\mathbbm{P}\Big(\frac{X^{r}}{\varepsilon^{r}}\geq i\Big)}\leq  &\sum_{i=1}^{\infty}{\int_{i-1}^{i}\mathbbm{P}\Big(\frac{X^{r}}{\varepsilon^{r}}\geq t \Big) dt}\notag\\
=&\int_{0}^{\infty}{\mathbbm{P}\Big(\frac{X^{r}}{\varepsilon^{r}}\geq t\Big)} dt=\frac{\mathbbm{E}(X^{r})}{\varepsilon^{r}}< \infty.
\end{align}
%where we assume that $\varepsilon^{-1}$ is an integer, otherwise we apply the arguments above on a smaller value of $\tilde{\varepsilon}<\varepsilon$, where ${\tilde{\varepsilon}}^{-1}$ is integer-valued, to obtain the same conclusion, since $\sum_{i=1}^{\infty}{\mathbbm{P}}( \frac{X_{i}}{i^{1/r}}\geq \varepsilon )\leq \sum_{i=1}^{\infty}{\mathbbm{P}}( \frac{X_{i}}{i^{1/r}}\geq \tilde{\varepsilon})<\infty$.
According to Borel-Cantelli Lemma~\cite{Durrett:2010:PTE:1869916}, we know that
\begin{align}\label{limsupequ}
\limsup \Big\{\frac{X_{i}}{i^{1/r}}\Big\}=0	
\end{align}
holds almost surely.
%$\limsup \Big\{\frac{X_{i}}{i^{1/r}}\Big\}=0$.
We next show the following lower bound on the limit superior of sequence $\{\frac{X_{i}}{i^{1/r}}\}$,
%Combined with (\ref{limsupequ}), we can see that
\begin{align}\label{maximumofnvari}
\lim_{n\rightarrow \infty}\frac{\max_{i\in [n]}{X_{i}}}{n^{1/r}}=\lim_{n\rightarrow \infty}\frac{\max_{i\in [n+1, 2n]}{X_{i}}}{n^{1/r}} \leq 2^{1/r}	\cdot \limsup \Big\{\frac{X_{i}}{i^{1/r}}\Big\} =0, w.p.1,
\end{align}
where the first equality holds due to fact that $X_{i}$ are identically distributed.

On the other hand, when $\mathbbm{E}[X^{r}]=\infty$,
\begin{align}
\sum_{i=1}^{\infty}{\mathbbm{P}}\Big( \frac{X_{i}}{i^{1/r}}\geq M\Big)=\sum_{i=1}^{\infty}{\mathbbm{P}\Big(\frac{X^{r}}{M^{r}}\geq i\Big)}\geq  &\sum_{i=1}^{\infty}{\int_{i}^{i+1}\mathbbm{P}\Big(\frac{X^{r}}{M^{r}}\geq t \Big) dt}\geq \frac{\mathbbm{E}(X^{r})}{M^{r}}-1= \infty
\end{align}
holds for any $M>0$. Since events $\{\frac{X_{k}}{i^{1/r}}\geq M\}_{k\in [n]}$ are independent, we are able to conclude that $ \mathbbm{P}(\frac{X_{n}}{n^{1/r}}\geq M, \io)=1 \;(\forall M>0)$ according to the second Borel-Cantelli Lemma~\cite{Durrett:2010:PTE:1869916}, and thus $\lim_{n\rightarrow \infty}\frac{\max_{i\in [n]}{X_{i}}}{n^{1/r}}=\infty$ almost surely. The proof of the first case is complete.

%\paragraph{Convergence Speed.}
\paragraph{Case $2$.} With the assumption of bounded $(r+\epsilon)$-th moment of random variables $\{X_{i}\}_{i\in [n]}$, we can obtain the following bound  when $n\geq N_{\varepsilon,\delta}= \Big(\frac{\sup \mathbbm{E}[X^{r+\epsilon}]}{\delta \cdot \varepsilon^{r+\epsilon}} \Big)^{\frac{r}{\epsilon}}$,
\begin{align*}
\mathbbm{P}\Big({\frac{\max_{i\in [n]}{X_{i}}}{{n}^{1/r}} \leq \varepsilon}\Big)=& 1-\mathbbm{P}\Big({\frac{\max_{i\in [n]}{X_{i}}}{n^{1/r}} > \varepsilon}\Big)\\
=& 1-\mathbbm{P}\Big( \bigcup_{i\in [n]} \Big\{\frac{X_{i}}{n^{1/r}}> \varepsilon \Big\} \Big)\\
\geq & 1- n\cdot \mathbbm{P}\Big(X^{r+\epsilon}> n^{1+\epsilon/r}\cdot \varepsilon^{r+\epsilon} \Big)\tag{union bound}\\
\geq & 1- \frac{\sup\mathbbm{E}(X^{r+\epsilon})}{N_{\varepsilon,\delta}^{\epsilon/r}\cdot \varepsilon^{r+\epsilon}}= 1-\delta,
\end{align*}
in which the second inequality follows from Markov inequality, \ie, $\mathbbm{P}(\frac{X_{i}}{n^{1/r}}> \varepsilon)=\mathbbm{P}(X^{r+\epsilon}> n^{1+\epsilon/r}\cdot \varepsilon^{r+\epsilon})\leq \frac{\mathbbm{E}(X^{r+\epsilon})}{n^{1+\epsilon/r}\cdot \varepsilon^{r+\epsilon}}$. The proof of the general case is complete.

%
%\begin{align}
%\lim_{n\rightarrow \infty}\Big\{\frac{\max_{i\in [n]}{X_{i}}}{n^{1/r}}\Big\}\leq \limsup \Big\{\frac{X_{i}}{i^{1/r}}\Big\}
%\end{align}
\end{proof}

\section{Proof of Observation~\ref{observationofflowtime}}\label{appendixob}
\begin{proof}
We first show the lower bound of total arrival time. Indeed the total arrival time satisfies that
\begin{align}\label{totalarrivaltime}
\sum\nolimits_{i\in [n]}{r_{i}}=\sum\nolimits_{i\in [n]}{\sum\nolimits_{j\in [i]}{\Delta r_{j}}}=\sum\nolimits_{i\in [n]}{(n-i)\cdot \Delta r_{i}}\geq \frac{n\cdot \sum\nolimits_{i\in [\frac{n}{2}]}{\Delta r_{i}}}{2},
\end{align}
for which we know that $\lim_{n\rightarrow +\infty}{\frac{\sum_{i\in [\frac{n}{2}]}{[\Delta r_{i}-\mu^{(i)}_{r}]}}{n}}=0$ w.p.1, since
\begin{align}\label{totalarrivalpro}
\mathbbm{P}\Big(\frac{\sum_{i\in [\frac{n}{2}]}[\Delta r_{i}-\mu^{(i)}_{r}]}{n}\geq \varepsilon\Big)\leq \frac{\sum_{i\in [\frac{n}{2}]}{\sigma^{(i)}_{r}}^{2}}{n^{2}\varepsilon^{2}}\rightarrow 0,
\end{align}
where $\sigma^{(i)}_{r}$ represents the variance of the $k$-th interarrival time. Consequently we know that
$\lim_{n\rightarrow +\infty}{\sum_{i\in [n]}{r_{i}}}=\Omega(n^2)$, and the following equality holds if the total flow time under the optimal algorithm $\pi^{*}$ is in the order of $o(n^{2}/B^{(n)})$,
\begin{align*}
\lim_{n\rightarrow \infty}\frac{\sum_{i\in [n]}{C^{\pi}_{i}}}{\sum_{i\in [n]}C^{\pi^{*}}_{i}}	=&\lim_{n\rightarrow \infty}\frac{1+(\sum_{i\in [n]}{f^{\pi}_{i}})/(\sum_{i\in [n]}{r_{i}})}{1+(\sum_{i\in [n]}{f^{*}_{i}})/(\sum_{i\in [n]}{r_{i}})}\\
=&\lim_{n\rightarrow \infty}\frac{1+\frac{o(B^{(n)}\cdot (n^{2}/B^{(n)}))}{\Omega(n^{2})}}{1+\frac{o(n^{2})}{\Omega(n^{2})}}=1.\tag{Theorem~\ref{tightuppinfworkcon}}
\end{align*}
%Alternatively, we can also show the $\Omega(n^{2})$ lower bound from the job workload point of view. Let $\sigma_{k}$ be the $k$-th maximum
%on the optimal total completion time
In other words, the completion time is indeed dominated by the total arrival time. The proof is complete.
 \end{proof}

\section{Proof of Proposition~\ref{singlestablepro}} \label{appendixpfstab}
\begin{proof}
Note that for $t\in [r_{k}, r_{k+1})$, the remaining workload under any work-conserving algorithm satisfies the Lindley equation~\cite{asmussen2008applied},
\begin{align*}
W(t)=(W(r_{k})+p_{k}-(t- r_{k}))^{+}, \;\forall  t\in [r_{k}, r_{k+1}),
\end{align*}
in which $W(r_{k})=W_{k}$ equals to the waiting time of the $k$-th arriving job under FCFS discipline.

Similar as the proof of Lemma~\ref{singleopt}, we have
\begin{align*}
&\lim_{n\rightarrow \infty}{\frac{W(r_{n})}{n^{1/2+\epsilon}}}=\lim_{n\rightarrow \infty}\frac{W_{n}}{n^{1/2+\epsilon}}\\
=&\lim_{n\rightarrow \infty}{\frac{T_{n}-\min_{k\in [n]}{T_{k}} }{n^{1/2+\epsilon}}}\\
=&\lim_{n\rightarrow \infty}{\Big\{\frac{[T_{n}-\sum_{k\in [n]}{\mu^{(k)}_{v}}]+[\max_{k\in [n]}{\{-T_{k}\}}-\sum_{k\in [n]}{(-\mu^{(k)}_{v})\cdot \mathbbm{1}_{\mu^{(k)}_{v}<0}}]+[\sum_{k\in [n]}{\mu^{(k)}_{v}\cdot \mathbbm{1}_{\mu^{(k)}_{v}>0}}]}{n^{1/2+\epsilon}}}\Big\}\\
=&o(1)+o(1)+\Omega(1)=\Omega(1), w.p.1,
\end{align*}
when condition (\ref{singleserverstable}) does not hold, which implies that $\mathbbm{P}(\lim_{t\rightarrow \infty}{W(t)}=\infty)=1$. Indeed we have
\begin{align*}
\lim_{t\rightarrow \infty}{\frac{W(t)}{t^{1/2}}}\geq \lim_{n\rightarrow \infty}{\frac{W(r_{\sigma_{t}})+p_{\sigma_{t}}-\Delta_{\sigma_{t}}}{\sigma^{1/2}_{t}}\cdot \Big(\frac{\sigma_{t}}{t}\Big)^{1/2}}=\Omega\Big( \frac{1}{\limsup_{n}{\mu^{(n)}_{r}} }\Big),
\end{align*}
where $\sigma_{t}=\min\{k\in [n]|r_{k}\geq t\}$ and
\begin{align*}
 \lim_{t\rightarrow \infty}\frac{\sigma_{t}}{t}\geq \lim_{t\rightarrow \infty}{\frac{\sigma_{t}}{r_{\sigma_{t}+1}}}\geq \frac{1}{\limsup_{n}{\mu^{(n)}_{r}} }, w.p.1.
 \end{align*}
We remark that it suffices to consider work-conserving disciplines, as it always incur the minimum possible remaining workload when there is a single server. The proof is complete.
\end{proof}

\section{Proof of Propostion~\ref{weightcompletiontheo}}\label{appendixpro}

\begin{proof}
We finish the proof by similar arguments as the unit weight case. Specifically, we first remark that the weighted arrival time is also in the order of $\sum_{i\in [n]}{\omega_{i}\cdot r_{i}}=\Omega(n^{2})$, as $\{\omega_{i}\cdot r_{i}\}_{i\in [n]}$ is a sequence of independent random variables with bounded mean value and variance,
% of $\mathbbm{E}[\omega_{i}\cdot r_{i}]\leq \sqrt{\mathbbm{E}[\omega^{2}_{i}]\cdot \mathbbm{E}[r^{2}_{i}]}=\sqrt{[\sigma_{\omega}+\mu^{2}_{\sigma}]\cdot[\sigma_{r}+\mu^{2}_{r}]}$.
Based on lemma~\ref{randommax}, we know that $\omega^{(n)}_{\max}=\max_{k\in [n]}{\omega_{k}}=o(n^{1/\alpha})$ and the weighted flow time satisfies that
\begin{align*}
\sum_{i\in [n]}{\omega_{i}\cdot f^{\pi}_{i}}\leq \omega_{\max}\cdot\sum_{i\in [n]}{f^{\pi}_{i}}\leq \omega_{\max}\cdot B^{(n)}\cdot \sum_{i\in [n]}f^{*}_{i}=O(\omega_{\max}\cdot B^{(n)}\cdot n)=O(n^{1/\alpha+1/\beta+1})=o(n^{2}).
\end{align*}
Consequently,
\begin{align}
\lim_{n\rightarrow \infty}\frac{\sum_{i\in [n]}{\omega_{i}\cdot C^{\pi}_{i}}}{\sum_{i\in [n]}\omega_{i} \cdot C^{\pi}_{i}}\leq &\lim_{n\rightarrow \infty} \Big(1+\frac{\sum_{i\in [n]}{\omega_{i} \cdot f^{\pi}_{i}}}{\sum_{i\in [n]}{\omega_{i} \cdot r_{i}}}\Big)=\lim_{n \rightarrow \infty} \Big(1+\frac{o(n^{2})}{\Omega(n^{2})}\Big)=1, w.p.1.
%\leq &\lim_{n\rightarrow \infty}1+\frac{o(B^{(n)}\cdot \omega_{\max} \cdot (n^{2}/(B^{(n)}\cdot \omega_{\max})))}{\Omega(n^{2})}=1,
\end{align}	
If we assume moreover that $\mathbbm{E}[e^{\epsilon \omega_{k}}]\leq \infty \;(\forall k\in [n])$ for some $\epsilon>0$, then
\begin{align*}
\mathbbm{P}(\omega^{(n)}_{\max}\geq \lambda \log n)=\mathbbm{P}(e^{\omega^{(n)}_{\max}}\geq n^{\lambda})\leq \frac{\mathbbm{E}[e^{\omega^{(n)}_{\max}}]}{n^{\lambda}}\leq \frac{\sum_{k\in [n]}\mathbbm{E}[e^{\omega_{k}}]}{n^{\lambda}}\leq \frac{\sup \mathbbm{E}[e^{\omega_{n}}]}{n^{\lambda-1}}\rightarrow 0
\end{align*}
\end{proof}

\section{Proof for relaxing independence assumption}\label{appendfact}
\begin{proof}
To see this fact, (\ref{totalarrivalpro}) and (\ref{sigma1ineq}) follows from the same reasoning using Chebyshev inequality,  (\ref{kolmo}) can be proved by applying Fact~\ref{gencondition}, a generalized Kolmogorov's inequality on $X_{k}=(p_{k}-\Delta_{k})-(\mu^{(k)}_{p}-\mu^{(k)}_{r})$ for $\alpha=2$, due to the independence of workload and interarrival time.

A special case is when the job workload and interarrival time are \emph{negative correlated}, \ie,
\begin{align}\label{negativecor}
\cov(p_{\ell}, p_{k}), \cov(\Delta r_{\ell}, \Delta r_{k})<0, \forall \ell\neq k \in [n],
\end{align}
under which the correctness of inequality (\ref{gencon}) is guarteened for $u^{(k)}_{p}=\var[p_{k}]$, $u^{(k)}_{r}=\var[\Delta r_{k}]$ since
\begin{align*}
\mathbbm{E}\Big[\sum^{j}_{k=i}{(p_{k}-\mu^{(k)}_{p})}\Big]^{2}=\sum^{j}_{k=i}\mathbbm{E}[(p_{k}-\mu^{(k)}_{p})^{2}]+\sum_{i\leq \ell\neq k\leq j}\mathbbm{E}[(p_{\ell}-\mu^{(\ell)}_{p})(p_{k}-\mu^{(k)}_{p})]\leq \sum^{j}_{k=i}\mathbbm{E}[(p_{k}-\mu^{(k)}_{p})^{2}].
\end{align*}
\end{proof}

\section{Generalization of Kolmogorov's inequality}
\begin{fact}[\cite{billingsley2013convergence}]
\label{genkol}
Suppose $X_{1}, X_{2},\ldots, X_{n}$ are random variables with $\mathbbm{E}[X_{i}]=0$ and there exists sequence $\{u_{\ell}\}_{\ell\in [n]}$ such that
\begin{align}\label{gencondition}
% \mathbbm{E}\Big[\Big|S_{j}-S_{i})\Big|\Big]
\mathbbm{E}[|S_{j}-S_{i}|^{\alpha}]\leq \sum_{i\leq \ell\leq j}{u_{\ell}}, 1\leq i\leq j\leq n,
 \end{align}
 where $S_{\ell}=\sum_{i\in[\ell]}{X_{i}}$, then there exists constant $K$ for all positive $\lambda$,
 \begin{align*}
 \mathbbm{P}(\max_{\ell\in [n]}|S_{\ell}|\geq \lambda)\leq K\cdot \frac{\sum_{\ell\in [n]}{u_{k}}}{\lambda^{\alpha}}
 \end{align*}
 %There exists $K$
 \end{fact}

\section{Remark on the Stability Condition}
We remark that the lower bound on the optimal flow time is non-trivial under the stability condition, \ie, the existence of stationary distribution. For example, consider independent random variable sequence $\{a^{(n)}_{k}\}_{k\in [n]}$ where $a_{k}=M^{(n)}_{1}$ with probability $p_{k}=\frac{1}{(k+1)^{2}}$, otherwise $a_{k}\in [0,M^{(n)}_{2}]$, where $M^{(n)}_{2}<M^{(n)}_{1}$ is finite. Then
\begin{align}\label{exampleineq}
\mathbbm{P}(\sum_{k\in[n]}{a_{k}}<M^{(n)}_{1})\leq &\mathbbm{P}(a_{k}<M^{(n)}_{1},\forall k\in [n])\\
=&\prod_{k\in [n]}\mathbbm{P}(a_{k}\in [0, M^{(n)}_{2}],\forall k\in [n])=\prod_{k\in [n]}(1-p_{k})=\frac{n+2}{2n+2}.
\end{align}
Let $M^{(n)}_{1}\rightarrow \infty$, from (\ref{exampleineq}) we know that ${\mathbbm{P}(\lim_{n\rightarrow \infty}\sum_{k\in[n]}{a_{k}}=\infty)}\geq{\mathbbm{P}(\lim_{n\rightarrow \infty}\sum_{k\in[n]}{a_{k}}\geq M^{(n)}_{1})}\geq 1/4$, though $\lim_{n\rightarrow \infty}\mathbbm{P}(a_{n}<\infty)=1$.
\end{document}